%% file: SOFSEM_version.tex
\documentclass[runningheads]{llncs}

\usepackage[T1]{fontenc}
\usepackage[utf8]{inputenc}
\usepackage{graphicx}
\usepackage{amsmath, amssymb, amsfonts}
\usepackage[hidelinks]{hyperref}
\usepackage{tikz}
\usepackage{cite}
\usetikzlibrary{decorations.pathreplacing,arrows.meta,shapes.geometric,decorations.markings}
\usepackage[ruled,vlined]{algorithm2e}
\tikzstyle{vertex}=[circle, draw, inner sep=0pt, minimum size=5pt]

\usepackage{orcidlink}
\DeclareMathOperator{\dist}{dist}
\DeclareMathOperator{\diam}{diam}
\usepackage{tikz}
\usetikzlibrary{arrows.meta,positioning,shapes.geometric}
\usepackage{subcaption}   

\newcommand{\proofpara}[1]{\medskip\noindent\textit{#1.}\ }
\usepackage{hyperref}
\usepackage{xcolor}

\spnewtheorem{observation}{Observation}{\bfseries}{\itshape}
\spnewtheorem{reduction}{Reduction Rule}{\bfseries}{\itshape}
\spnewtheorem{branching}{Branching Rule}{\bfseries}{\itshape}
\spnewtheorem{red}{Reduction Rule}{\bfseries}{\itshape}
\spnewtheorem{Claim}{Claim}{\bfseries}{\itshape}

\begin{document}
\title{Kernelization of $2$-Club Cluster Edge Deletion on Interval Graphs}
\titlerunning{$2$-Club Cluster Edge Deletion}
%
\author{Ajinkya Gaikwad}
\authorrunning{A.\,Gaikwad}
%
\institute{Faculty of Information Technology, Czech Technical
University in Prague, 
\\ Czech Republic
\email{\texttt{ajinkya.gaikwad@fit.cvut.cz}}}

\maketitle              
\begin{abstract}
The \emph{$s$-Club Cluster Edge Deletion} problem asks whether, given a graph
$G$ and an integer $k$, one can delete at most $k$ edges so that every remaining
connected component has diameter at most~$s$. This generalizes the classical
\emph{Cluster Edge Deletion} problem by permitting components of bounded
diameter instead of requiring cliques. On general graphs, $2$-Club Cluster Edge Deletion is known
to be fixed-parameter tractable when parameterized by $k$, but it remains open whether it admits a polynomial kernel, as posed in~\cite{ABUKHZAM2023113864}. 
Motivated by this question, we study the problem on interval graphs and obtain a
polynomial vertex kernel of size $\mathcal{O}(k^{5})$. As a complementary
result, we also show that the \emph{$s$-Club Cluster Edge Deletion} problem is polynomial time solvable on unit interval graphs. We also show that $2$-Club Cluster Edge Deletion is NP-hard even on split graphs.
\end{abstract}

\section{Introduction}

Clustering problems under structural constraints form a central topic in
algorithmic graph theory, with applications ranging from social network
analysis~\cite{doi:10.1073/pnas.122653799,FORTUNATO201075} to bioinformatics~\cite{HARTUV2000249,Alizadeh1995} and data mining~\cite{Berkhin2006}.
A classical example is \emph{Cluster Edge Deletion}, where the goal is to delete
at most $k$ edges from a graph so that every connected component becomes a
clique.
It is known that this problem is NP-hard on general graphs~\cite{SHAMIR2004173}, and determining their computational complexity even on restricted graph classes has been the subject of extensive research~\cite{BONOMO2015600,BONOMO201559,doi:10.1142/S0129054107004656,GAO20132763}.
This \emph{Cluster Edge Deletion} problem and its variants are well studied from the perspectives of exact algorithms~\cite{Bocker2011Exact},
parameterized complexity~\cite{10.1007/3-540-44849-7_17,BOCKER20095467,cao_et_al:LIPIcs.IPEC.2021.13,TSUR2022106171,10.1007/978-3-642-16926-7_17,Huffner2010ClusterVertexDeletion}, and kernelization~\cite{GUO2009718,10.1007/978-3-642-17493-3_8,cao_et_al:LIPIcs.IPEC.2021.13}.
These problems have found applications in bioinformatics, software security (vulnerability assessment), and text classification~\cite{10.1007/11847250_2,fadiel2006computational,9253144,doi:10.1142/S1793351X20500087,9030927}.
A natural generalization replaces cliques by connected subgraphs of bounded
diameter.
For a fixed integer $s \ge 1$, an \emph{$s$-club} is a graph of diameter at most
$s$, and the \emph{$s$-Club Cluster Edge Deletion} problem asks whether a given
graph can be transformed, by deleting at most $k$ edges, into a disjoint union
of $s$-clubs.
For $s=1$, the problem coincides with Cluster Edge Deletion.
For larger values of $s$, the problem allows more flexibility inside clusters
while still enforcing a strong global structure.

A number of graph modification problems aim at transforming a graph into a disjoint union of $2$-clubs, most notably \textsc{2-Club Cluster Vertex Deletion}, \textsc{2-Club Cluster Edge Deletion}, and \textsc{2-Club Cluster Editing}. Each of these variants is NP-complete \cite{10.1007/978-3-642-29700-7_22}.
Moreover, it was shown in \cite{BAZGAN2025247} that finding the minimum number of edges to add to a split graph in order to obtain diameter at most~2 is W[2]-hard when parameterized by the number of added edges, and \cite{10.1007/978-3-030-75242-2_15} established that 2-Club Cluster Editing is W[2]-hard with respect to the number of modified edges.
On general graphs, $2$-Club Cluster Edge Deletion is known to be
fixed-parameter tractable when parameterized by $k$~\cite{ABUKHZAM2023113864}.
Abu-Khzam et al.~\cite{ABUKHZAM2023113864} presented the currently fastest parameterized algorithm, running in $\mathcal{O}^*(2.692^{k})$ time.
However, despite this progress, it remains open whether the problem admits a
polynomial size kernel on general graphs.
This motivates a systematic study of the problem on restricted graph classes,
where additional structure can be exploited to obtain stronger algorithmic
results.
In this work, we focus on two fundamental subclasses of chordal graphs:
\emph{split graphs} and \emph{interval graphs}.
Both classes arise naturally in applications and admit rich structural
characterizations.
Split graphs are graphs whose vertex set can be partitioned into a clique and an independent set, while interval graphs are intersection graphs of intervals on the real line.
These representations allow for geometric reasoning and strong ordering
properties that are unavailable in general graphs.

\section{Preliminaries}

We consider finite, simple, undirected graphs.
For a graph $G$, we denote its vertex set by $V(G)$ and its edge set by
$E(G)$.
For a vertex $v\in V(G)$, the \emph{open neighborhood} of $v$ is
$N_G(v)=\{u\in V(G)\mid uv\in E(G)\}$, and the \emph{closed neighborhood} is
$N_G[v]=N_G(v)\cup\{v\}$.
For a set $X\subseteq V(G)$, we write $G[X]$ for the subgraph of $G$ induced
by~$X$.
The \emph{distance} between two vertices $u,v \in V(G)$, denoted by $\dist_G(u,v)$, is the length of a shortest path between $u$ and $v$.
The \emph{diameter} of a connected graph $H$ is $\diam(H) = \max\{\dist_H(u,v) \mid u,v \in V(H)\}$.

\begin{definition}\cite{Foldes1977SplitGraphs}
A split graph is one whose vertex set can be partitioned as the disjoint union of an independent set $I$ and a clique $C$(either of which may be empty).
\end{definition}

Split graphs are chordal and perfect, and a split partition, that is, a partition of the vertex set into a clique and an independent set, can be found in polynomial time~\cite{Foldes1977SplitGraphs}.

\begin{definition}
A graph $G$ is an \emph{interval graph} if there exists a family of
intervals $\{I_v=[\ell_v,r_v]\mid v\in V(G)\}$ on the real line such that
$uv\in E(G)$ if and only if $I_u\cap I_v\neq\emptyset$.
Such a family is called an \emph{interval representation} of $G$.
\end{definition}

Interval graphs admit a linear ordering of vertices by increasing left
endpoints of their intervals~\cite{10.5555/984029}.
Throughout the paper, when working with an interval graph, we assume a
fixed interval representation and use the corresponding ordering
implicitly.

\section{On Split Graphs}

For $s=1$, an $s$-club is precisely a clique, since a connected graph has
diameter at most $1$ if and only if every two distinct vertices are adjacent.
Hence, $1$-\textsc{Club Cluster Edge Deletion} coincides with the classical
\textsc{Cluster Edge Deletion} problem. Since \textsc{Cluster Edge Deletion}
is solvable in polynomial time on split graphs~\cite{konstantinidis_et_al:LIPIcs.MFCS.2019.12}, the case $s=1$
is polynomial-time solvable on split graphs.
Moreover, since every connected component of a split graph has diameter at most~3,
the problem is also polynomial-time solvable on split graphs when $s\geq3$.
In contrast, we show that {\sc $s$-Club Cluster Edge Deletion} becomes NP-hard on split graphs already for $s=2$.
In particular, we prove the following theorem.

\begin{theorem}\label{split NP-hard}
{\sc $2$-Club Cluster Edge Deletion} is NP-hard on split graphs.
\end{theorem}

\section{On Interval Graphs}

In this section, we present a polynomial kernel for \textsc{$2$-Club Cluster Edge Deletion} on interval graphs. 
The kernelization proceeds by applying a sequence of reduction rules that exploit structural properties of interval graphs to bound both the number and size of connected components.

\proofpara{Scope of the Reduction Rules} We first clarify the scope of the reduction rules developed in this section. 
Reduction Rules~7--10 apply to the general \textsc{$s$-Club Cluster Edge Deletion} problem on interval graphs and hold for all $s \ge 2$. 
These rules rely only on global diameter arguments, interval ordering properties, and domination relations, and are independent of the exact value of $s$.

In contrast, Reduction Rule~11 is specific to the case $s = 2$. 
Its correctness crucially exploits structural properties of $2$-clubs, most notably that every pair of vertices must have a common neighbor. 
This allows us to identify vertex pairs whose distance is forced to be at least $3$ under any feasible edge-deletion set, leading to both safe reductions and no-instance detection.

For larger values of $s$, such local common-neighbor arguments are no longer sufficient, as more complex distance interactions arise in interval graphs. 
Consequently, extending Reduction Rule~11 to arbitrary $s$ remains an open problem and appears to require new structural insights.
\medskip

\proofpara{Computational Status}
We also emphasize that, despite the kernelization results obtained in this section,
the computational complexity of \textsc{$2$-Club Cluster Edge Deletion} on interval graphs
is not yet fully understood. In particular, it remains open whether the problem is
NP-hard on interval graphs.

Furthermore, the difference between diameter one and diameter two is known to be
substantial even from a parameterized perspective. Misra, Panolan, and Saurabh
showed that a subexponential-time algorithmic phenomenon that holds for \textsc{Cluster Edge Deletion}
(the case $s=1$) do not extend to $s$-\textsc{Club Cluster Edge Deletion} once
$s \ge 2$. In addition, while \textsc{Cluster Edge Deletion} is polynomial-time
solvable on split graphs, \textsc{$2$-Club Cluster Edge Deletion} is already
NP-hard on the same graph class. Consequently, the polynomial-time solvability
of \textsc{Cluster Edge Deletion} on interval graphs \cite{konstantinidis_et_al:LIPIcs.MFCS.2019.12} should not be viewed as strong
evidence that \textsc{$2$-Club Cluster Edge Deletion} is polynomial-time solvable
on interval graphs.

Thus, the present work establishes polynomial kernelization under the parameter
$k$, but leaves open the fundamental question of whether the problem admits a
polynomial-time algorithm on interval graphs or is NP-complete even on this
restricted graph class.

\subsection{Overview of the Kernelization Strategy}

We describe the kernelization in a sequence of conceptually distinct stages.

\medskip
\noindent\textbf{Stage 1: Bounding diameter and number of components.}
We begin by applying Reduction Rule~7 exhaustively. This removes every connected
component $H$ with $\diam(H)\le s$, since such a component already forms a valid
$s$-club, and rejects the instance if some component has diameter exceeding
$(s+1)(k+1)$. Consequently, every remaining component has diameter bounded by
$O(k)$.

Next, we apply Reduction Rule~8, which ensures that the number of connected
components is at most $k$. Indeed, each remaining component requires at least
one edge deletion to reduce its diameter to at most $s$, and hence more than $k$
components would exceed the deletion budget.

\medskip
\noindent\textbf{Stage 2: Finding large clique or independent set in a component.}
It remains to bound the size of each connected component $H$. Suppose that
$H$ is large. Since interval graphs are perfect~\cite{984029},
we have $\chi(H)=\omega(H)$. Thus, in any proper coloring of $H$, one color
class contains at least $|V(H)|/\omega(H)$ vertices, implying that
\[
\alpha(H)\ge \frac{|V(H)|}{\omega(H)}.
\]
Hence,
$\alpha(H)\cdot\omega(H)\ge |V(H)|$,
and therefore
$\max\{\omega(H),\alpha(H)\}\ge \sqrt{|V(H)|}$.
Consequently, $H$ contains either a large independent set or a large clique.
We branch on these two possibilities.

\medskip
\noindent\textbf{Stage 3: Large independent set case.}
If $H$ contains a large independent set, then Lemma~7 implies that there exists
a vertex whose neighbors include a long consecutive block of vertices in the
interval ordering of the independent set. Reduction Rule~9 exploits this
structure by deleting a carefully chosen vertex from this block. The safety of
this rule follows from the fact that the remaining vertices in the block can
simulate the role of the deleted vertex in all relevant distance arguments.

\medskip
\noindent\textbf{Stage 4: Large clique case.}
If $H$ contains a large clique, then Lemma~9 extracts from it either a large
nested clique or a large staircase clique.

In the nested case, Reduction Rule~10 applies: a suitably chosen middle vertex
can be deleted because its neighborhood is dominated by surrounding vertices in
the nested structure.

In the staircase case, Reduction Rule~11 applies, which is specific to the case
$s=2$. Here, either one of two central vertices can be safely deleted, or the
rule correctly identifies a no-instance by exhibiting a pair of vertices that
cannot be kept within distance two without exceeding the deletion budget.

\medskip
\noindent\textbf{Stage 5: Termination and kernel size.}
Each of the above rules strictly reduces the size of the instance or rejects it.
Therefore, as long as some connected component exceeds the prescribed size
threshold, one of the reduction rules applies. Once no rule is applicable, every
remaining component has size bounded by a polynomial in $k$. Since there are at
most $k$ components this yields the
desired polynomial kernel.

\subsection{Kernelization Algorithm}

We now present the reduction rules formally.

\medskip
\noindent\textbf{Stage 1: Bounding diameter and number of components.} In this stage, we bound the diameter and number of connected components in the input graph. We begin with a reduction rule that helps us bound the diameter of each connected component.

\begin{red}\label{rr:small-large-components}
Let $H$ be any connected component of $G$.
\begin{itemize}
    \item If $\diam(H)\le s$, remove $H$ without changing $k$.
    \item If $\diam(H) > (s+1)(k+1)$, then return a no-instance.
\end{itemize}
\end{red}

\begin{lemma}
Reduction Rule~\ref{rr:small-large-components} is safe.
\end{lemma}

\noindent Now, we provide a reduction rule that bounds the number of connected components in the input graph.

\begin{red}\label{rr:cc-count}
If $G$ has more than $k$ connected components, then return a no-instance.
\end{red}

\begin{lemma}
Reduction Rule~\ref{rr:cc-count} is safe.
\end{lemma}

\noindent At this point, Stage~1 is complete. We have removed all trivial components and ensured that every remaining connected component $H$ satisfies
$s < \diam(H) \le (s+1)(k+1)$,
and that the total number of connected components is at most $k$.

\medskip
\noindent\textbf{Stage 2: Finding large clique or independent set in a component.} The following lemma shows that any sufficiently large component contains either a large independent set or a large clique, and such a structure can be found in polynomial time.

\begin{lemma}\label{lem:large-structure}
Let $H$ be a connected component of an interval graph. 
If  $|V(H)| > (6k+7)^4$, then $H$ contains either an independent set or a clique of size at least $(6k+7)^2$.
Moreover, such a structure can be found in polynomial time using the interval
representation of $H$.
\end{lemma}

\begin{proof}
It is well known that for interval graphs,
$\max\{\omega(H), \alpha(H)\} \ge \sqrt{|V(H)|}$.
Thus, if $|V(H)| > (6k+7)^4$, we obtain
$\max\{\omega(H), \alpha(H)\} > (6k+7)^2$,
which implies that $H$ contains either a clique or an independent set of size at least $(6k+7)^2$.
\end{proof}

\noindent This completes Stage~2. We can now assume that every remaining component $H$ with $|V(H)|>(6k+7)^4$ contains either a large independent set or a large clique, which will be handled in the next stage.

\medskip
\noindent\textbf{Stage 3: Large independent set case.} We first consider the case when $H$ contains a large independent set. 
The following lemma shows that such a structure contains a long consecutive block in the interval ordering, which can be used to safely apply a reduction.

\begin{lemma}\label{lem:IS-long-consecutive-neighbour}
Let $H$ be a connected interval graph with $\diam(H) \leq (s+1)(k+1)$, and let 
$I = \{v_1,\ldots,v_{|I|}\} \subseteq V(H)$ be an independent set, where the vertices are ordered by nondecreasing left endpoints.
If $|I| \ge (s+1)(k+1)(6k+7)$, then there exists a vertex $u \in V(H)$ that is adjacent to at least $6k+7$ vertices of $I$, and these neighbours form a consecutive block
$v_a, v_{a+1}, \ldots, v_{a+6k+6}$ in the ordering of $I$.
\end{lemma}

\noindent We now use the structure identified in Lemma~\ref{lem:IS-long-consecutive-neighbour} to apply a reduction rule. 

\begin{red} \label{rr:IS-neighbour-irrelevant}
Let $H$ be a connected component of an interval graph, and let
$I=\{v_1,\ldots,v_{|I|}\}\subseteq V(H)$ be an independent set ordered by nondecreasing left endpoints.
Suppose there exists a vertex $u\in V(H)$ that is adjacent to a block of $6k+7$ consecutive vertices
$v_a, v_{a+1}, \ldots, v_{a+6k+6}$ of $I$.
Delete the vertex $v_{a+2k+2}$ and keep $k$ unchanged.
\end{red}

The following structural properties will be used in the proof of
Reduction Rule~\ref{rr:IS-neighbour-irrelevant}.

\begin{observation}[Path-interval union] Let \(P=(p_0,p_1,\ldots,p_\ell)\) be a path in a subgraph of an interval graph \(G\), under a fixed interval representation of \(G\). Then, the union of intervals $\bigcup_{j=0}^{\ell} I_{p_j}$  is a connected interval of the real line. \end{observation}

\begin{observation}\label{obs:interval-intersection-cross-edges} Let $G$ be an interval graph with fixed interval representation $\{I_v=[\ell_v,r_v] : v\in V(G)\}$. Let $F$ be an edge-deletion set, and let $C_1$ and $C_2$ be two distinct connected components of $G-F$. For $i\in\{1,2\}$, let  $J(C_i):=\bigcup_{w\in C_i} I_w $.  Then each $J(C_i)$ is a connected interval. Let  $J:=J(C_1)\cap J(C_2)$  and define  $X:=\{\,v\in V(G): I_v\cap J\neq\emptyset\,\}$.  Then every vertex $v\in X$ is incident to at least one edge of $F$ whose endpoints lie in two distinct connected components of $G-F$. \end{observation} 

\noindent Using Observation~\ref{obs:interval-intersection-cross-edges}, we now show that any component intersecting $I_v$ must lie in a restricted region.

\begin{Claim}\label{claim:middle-interval-squeezed}
Let $v := v_{a+2k+2}$. Let $x \in V(G)$ be a vertex such that
$I_x \cap I_v \neq \emptyset$, and let $C_x$ be the connected component of
$G-F$ containing $x$. Suppose that $C_x \neq C_u$. Then
$J(C_x) \subset \bigl(r_{v_{a+k+1}},\, \ell_{v_{a+3k+5}}\bigr)$.
Equivalently, every vertex in $C_x$ has its interval strictly contained between
$r_{v_{a+k+1}}$ and $\ell_{v_{a+3k+5}}$.
\end{Claim}

\noindent Using the structural restriction from Claim~\ref{claim:middle-interval-squeezed},
we now show that such components can be safely merged with $C_u$ without
violating the $s$-club property.

\begin{Claim}\label{claim:merge-clubs} Let \(F\) be an inclusion-minimal feasible edge-deletion set such that \(G-F\) is a cluster of \(s\)-clubs. Let \(v:=v_{a+2k+2}\), and let \(C_u\) denote the connected component of \(G-F\) containing \(u\). Let \(x\in V(G)\) be a vertex such that \(I_x\cap I_v\neq\emptyset\), and let \(C_x\) be the connected component of \(G-F\) containing \(x\). If \(C_x\neq C_u\), then the graph obtained from \(G-F\) by restoring all edges between \(C_u\) and \(C_x\) has \(C_u\cup C_x\) as an \(s\)-club. \end{Claim}

We are now ready to prove the safeness of Reduction Rule~\ref{rr:IS-neighbour-irrelevant}.

\begin{lemma}\label{lem:IS-neighbour-irrelevant}
    The Reduction Rule~\ref{rr:IS-neighbour-irrelevant} is safe.
\end{lemma}

\medskip
\noindent \textbf{Stage 4: Large clique case.}
We now conclude Stage~3. In the next stage, we consider the case where the
graph contains a sufficiently large clique and derive the corresponding
reduction rules.
We first introduce two canonical substructures of cliques in interval graphs,
which will be used to characterize large cliques.

\begin{definition}[Nested clique]\label{def:nested-clique}
Let $G$ be an interval graph with an interval representation
$\{I_v=[\ell_v,r_v] : v\in V(G)\}$.  
A clique $C=\{v_1,\ldots,v_\ell\}$ is called a \emph{nested clique} if all
intervals $I_{v_i}$ contain a common point $x$ and, after ordering the
vertices so that
$\ell_{v_1} \le \ell_{v_2} \le \cdots \le \ell_{v_\ell}$,
their right endpoints form a nonincreasing sequence:
$ r_{v_1} \ge r_{v_2} \ge \cdots \ge r_{v_\ell}$.
Equivalently, $I_{v_1} \supseteq I_{v_2} \supseteq \cdots \supseteq I_{v_\ell}$.
\end{definition}

\begin{definition}[Staircase clique]\label{def:staircase-clique}
Let $G$ be an interval graph with an interval representation
$\{I_v=[\ell_v,r_v] : v\in V(G)\}$.  
A clique $C=\{v_1,\ldots,v_\ell\}$ is called a \emph{staircase clique} if all
intervals $I_{v_i}$ contain a common point $x$ and, after ordering the
vertices so that
  $\ell_{v_1} \le \ell_{v_2} \le \cdots \le \ell_{v_\ell}$,
their right endpoints form a nondecreasing sequence:
$r_{v_1} \le r_{v_2} \le \cdots \le r_{v_\ell}$.
\end{definition}

We now show that every sufficiently large clique contains a large substructure
of one of these two types.

\begin{lemma}\label{lem:ES-clique-pattern-general}
Let $G$ be an interval graph and let $C$ be a clique of $G$ with
$|C|\ge \ell^{2}$, for some integer $\ell\ge 2$.
Then there exists a subset $S\subseteq C$ of size $\ell$ such that the intervals
$\{I_v : v\in S\}$ (under some fixed interval representation of $G$)
form one of the two configurations in
Figure~\ref{fig:es-interval-patterns-general}:
\begin{enumerate}
  \item a \emph{nested} pattern, i.e.\ their right endpoints form a
        nonincreasing sequence when ordered by left endpoints; or
  \item a \emph{staircase} pattern, i.e.\ their right endpoints form a
        nondecreasing sequence when ordered by left endpoints.
\end{enumerate}
In particular, $S$ is a clique of size $\ell$ contained in $C$ whose
intervals have one of these two forms.
\begin{figure}[ht]
  \centering
  \begin{subfigure}{0.48\textwidth}
    \centering
    \begin{tikzpicture}[scale=0.3]
      \draw[dashed] (0,0.5) -- (0,4.5);
      \node[above] at (0,4.5) {$x$};

      \foreach \y/\l/\r in {
        4/ -3/ 3,
        3.4/-2.5/ 2.6,
        2.8/-2/ 2.2,
        2.2/-1.5/ 1.7,
        1.6/-1/ 1.2,
        1.0/-0.5/ 0.7,
        0.4/ -0.2/ 0.2
      }{
        \draw[thick] (\l,\y) -- (\r,\y);
      }
    \end{tikzpicture}
    \caption{Nested pattern: right endpoints nonincreasing.}
  \end{subfigure}\hfill
  \begin{subfigure}{0.48\textwidth}
    \centering
    \begin{tikzpicture}[scale=0.3]
      \draw[dashed] (0,0.5) -- (0,4.5);
      \node[above] at (0,4.5) {$x$};

      \foreach \y/\l/\r in {
        4/ -3/ 0.6,
        3.4/-2.5/1.2,
        2.8/-2/ 1.8,
        2.2/-1.5/2.4,
        1.6/-1/ 3.0,
        1.0/-0.5/3.6,
        0.4/ 0/ 4.2
      }{
        \draw[thick] (\l,\y) -- (\r,\y);
      }
    \end{tikzpicture}
    \caption{Staircase pattern: right endpoints nondecreasing.}
  \end{subfigure}
  \caption{The two interval configurations on $\ell$ vertices guaranteed by
  Lemma~\ref{lem:ES-clique-pattern-general}.}
  \label{fig:es-interval-patterns-general}
\end{figure}
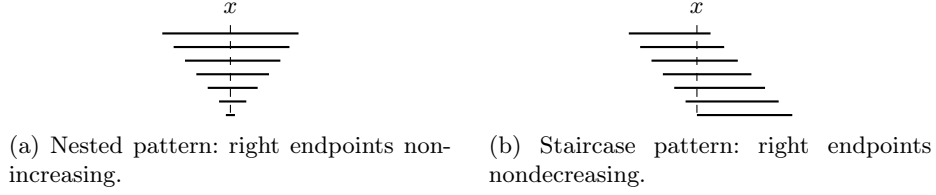
\end{lemma}
\begin{proof}
Fix an interval representation $\mathcal{I}=\{I_v=[\ell_v,r_v] : v\in V(G)\}$
of $G$. Since $C$ is a clique, the intervals $\{I_v : v\in C\}$ are pairwise
intersecting. By the Helly property for intervals~\cite{984029}, there exists a point $x$ that
is contained in every interval of $C$.
Order the vertices of $C$ as
$v_1,v_2,\ldots,v_m$,
where $m=|C|$, so that
\[
\ell_{v_1}\leq \ell_{v_2}\leq \cdots \leq \ell_{v_m}.
\]
Consider the sequence of right endpoints
$r_{v_1}, r_{v_2}, \ldots, r_{v_m}$.
Since $m\geq \ell^2$, by the Erdős--Szekeres monotone subsequence Theorem~\cite{Erdos1987},
this sequence contains a monotone subsequence of length $\ell$. Thus there
exist indices
$1\leq i_1<\cdots<i_\ell\leq m$
such that either
$r_{v_{i_1}}\geq r_{v_{i_2}}\geq \cdots \geq r_{v_{i_\ell}}$,
or
$r_{v_{i_1}}\leq r_{v_{i_2}}\leq \cdots \leq r_{v_{i_\ell}}$.
Let $S:=\{v_{i_1},v_{i_2},\ldots,v_{i_\ell}\}$.
Since $S\subseteq C$, the set $S$ is a clique. Moreover, the left endpoints of
the intervals in $S$ are ordered increasingly by the choice of the indices.

In the first case, the right endpoints are nonincreasing. Since all intervals
contain the common point $x$, the intervals form a nested pattern:
\[
I_{v_{i_1}}\supseteq I_{v_{i_2}}\supseteq \cdots \supseteq I_{v_{i_\ell}}.
\]
In the second case, the right endpoints are nondecreasing, and hence the
intervals form a staircase pattern.

Therefore, $C$ contains a subset $S$ of size $\ell$ whose intervals form either
a nested clique or a staircase clique. This proves the lemma.
\end{proof}

We now handle the nested clique case by showing how to safely delete a vertex.

\begin{red}\label{rr:nested-reduction}
Let $G$ be an interval graph and let 
$v_1,\ldots,v_{2k+3}$ be vertices that form a clique whose intervals satisfy
$I_{v_1} \supseteq I_{v_2} \supseteq \cdots \supseteq I_{v_{2k+3}}$
in some interval representation of $G$.   
Delete $v_{k+2}$ and keep $k$ unchanged.
\end{red}

\begin{lemma}
    Reduction Rule \ref{rr:nested-reduction} is safe.
\end{lemma}

We now handle the staircase clique case. Note that this reduction rule is given only for $s=2$ unlike the previous reduction rules.

\begin{red}\label{rr:staircase-pair-test}
Let $(G,k)$ be an instance of \textsc{$2$-Club Cluster Edge Deletion} and suppose that
$G$ contains a staircase clique
$C=\{v_1,\ldots,v_{2k+4}\}$, indexed so that
\[
\ell_{v_1} \le \ell_{v_2} \le \cdots \le \ell_{v_{2k+4}}
\quad\text{and}\quad
r_{v_1} \le r_{v_2} \le \cdots \le r_{v_{2k+4}}.
\]
in some interval representation of $G$.
Let $a:=v_{k+2}$ and $b:=v_{k+3}$.
If there exists a vertex $z\in\{a,b\}$ such that for all distinct
$x,y\in N_G(z)$ we have
\[
N_G(x)\cap N_G(y)\supsetneq \{z\},
\]
then delete $z$ and keep $k$ unchanged.
Otherwise, return a no-instance.
\end{red}

\begin{lemma}
    Reduction rule \ref{rr:staircase-pair-test} is safe.
\end{lemma}
\begin{proof}
We distinguish two cases depending on the outcome of the rule.

\medskip

\noindent \textbf{Case 1 (Deletion case).}
There exists a vertex $z \in \{a,b\}$ satisfying the pair condition.
In this case, the rule returns the instance $(G-z,k)$. We will prove the correctness of this case first.

Fix such a vertex $z\in\{a,b\}$ and let $G':=G-z$. Let $F\subseteq E(G)$ be an inclusion-minimal feasible solution with $|F|\le k$ such that $G-F$ is a cluster of $2$-clubs. Let $C_z$ be the connected component of $G-F$ that contains $z$. We first record two simple consequences. First, all vertices of the staircase clique $C$ belong to $C_z$. Indeed, separating the clique $C$ into two nonempty parts would require deleting at least $|C|-1\ge k+1$ edges, which exceeds the budget. Second, every vertex $x\in N_G(z)$ belongs to $C_z$. To see this, observe that, by the staircase structure, every neighbor of $z$ is adjacent to at least $k+1$ vertices of $C\setminus\{z\}$. 
Indeed, if $z=a$, then every neighbor of $z$ intersects either all of $v_1,\ldots,v_{k+1}$ or all of $v_{k+3},\ldots,v_{2k+4}$; if $z=b$, the analogous statement holds with the sets $v_1,\ldots,v_{k+2}$ and $v_{k+4},\ldots,v_{2k+4}$. In either case, the neighbor has at least $k+1$ neighbors in $C\setminus\{z\}$.
Since $C\setminus\{z\}\subseteq C_z$, if $x$ were not in $C_z$, then all these at least $k+1$ edges from $x$ to $C_z$ would have to be deleted, contradicting $|F|\le k$. Finally, by the inclusion-minimality of $F$, no edge of $F$ has both endpoints inside the same connected component of $G-F$. Indeed, adding back such an edge would not merge two components and could only decrease distances inside that component. Thus, in particular, every edge of $G$ with both endpoints in $C_z$ is present in $G-F$.

\smallskip\noindent
\emph{Forward direction.}
Assume $(G,k)$ is a yes-instance, and let $F$ be an inclusion-minimal solution as above. Consider $F':=F\cap E(G')$ and the graph $G'-F'$. Every connected component of $G-F$ distinct from $C_z$ is unchanged in $G'-F'$ and remains a $2$-club. It suffices to show that $C_z\setminus\{z\}$ is a $2$-club in $G'-F'$. 

Fix distinct vertices $p,q\in C_z\setminus\{z\}$. If $p$ and $q$ are adjacent in $G'-F'$, or if they have a common neighbor in $G'-F'$, then $\dist_{G'-F'}(p,q)\le 2$ and we are done. Suppose therefore that $p$ and $q$ are not adjacent in $G'-F'$ and have no common neighbor in $G'-F'$. Since $C_z$ is a $2$-club in $G-F$, the vertices $p$ and $q$ have distance at most two in $G-F$. Hence their only possible common neighbor in $G-F$ is $z$. In particular, $p,q\in N_G(z)$.
By the premise of the reduction rule, since $p,q\in N_G(z)$ are distinct, we have $N_G(p)\cap N_G(q)\supsetneq \{z\}$.  
Thus there exists a vertex $r\neq z$ adjacent to both $p$ and $q$ in $G$. Since $r$ is adjacent to $z$, the previous observation implies that $r\in C_z$. Hence $p,q,r$ all lie in the same component $C_z$ of $G-F$. By the inclusion-minimality of $F$, the edges $pr$ and $qr$ are not deleted. Therefore $r$ is a common neighbor of $p$ and $q$ in $G'-F'$. It follows that every pair of vertices in $C_z\setminus\{z\}$ has distance at most two in $G'-F'$. Hence $C_z\setminus\{z\}$ is a $2$-club, and therefore $G'-F'$ is a cluster of $2$-clubs. Thus $(G-z,k)$ is a yes-instance. 

\smallskip\noindent
\emph{Backward direction.}
Assume that $(G-z,k)$ is a yes-instance. Let $F'\subseteq E(G')$ with $|F'|\le k$ such that $G'-F'$ is a cluster of $2$-clubs, where $G':=G-z$. Let $C'_z$ be the connected component of $G'-F'$ that contains the clique $C\setminus\{z\}$. This component is well-defined, since $C\setminus\{z\}$ is a clique of size at least $2k+3$, and separating it into two nonempty parts would require more than $k$ edge deletions.

We first observe that every neighbor of $z$ belongs to $C'_z$. Indeed, by the staircase-structure observation above, every vertex adjacent to $z$ is adjacent to at least $k+1$ vertices of $C\setminus\{z\}$. Since all vertices of $C\setminus\{z\}$ lie in the same component $C'_z$ of $G'-F'$, any neighbor $x\in N_G(z)$ outside $C'_z$ would require deleting all these at least $k+1$ edges between $x$ and $C'_z$, contradicting $|F'|\le k$. Hence $N_G(z)\subseteq C'_z$. Now view $F'$ as an edge set of $G$. Since all neighbors of $z$ already lie in $C'_z$, adding $z$ back does not merge $C'_z$ with any other component of $G'-F'$. Thus the only component that changes is $C'_z$, which becomes $C_z:=C'_z\cup\{z\}$. 

It remains to show that $C_z$ is a $2$-club. Distances inside $C'_z$ do not increase when adding $z$. Therefore, it is enough to show that every vertex $p\in C'_z$ satisfies $\dist_{G-F'}(p,z)\le 2$. If $p\in N_G(z)$, then $\dist_{G-F'}(p,z)=1$. Hence assume that $p\notin N_G(z)$. Then $I_p$ is disjoint from $I_z$, and so $I_p$ lies completely either to the left or to the right of $I_z$. Assume first that $I_p$ lies to the left of $I_z$, that is, $r_p<\ell_z$. Let $x:=v_{2k+4}$. Since $x\in C'_z$ and $C'_z$ is a $2$-club in $G'-F'$, we have $\dist_{G'-F'}(p,x)\le 2$. Moreover, $p$ is not adjacent to $x$, so there exists a vertex $w\in C'_z$ adjacent to both $p$ and $x$ in $G'-F'$. Since $w$ is adjacent to $p$, we have $\ell_w\le r_p<\ell_z$. Since $w$ is adjacent to $x$, we have $r_w\ge \ell_x$. By the staircase ordering, and since $x$ and $z$ belong to the clique $C$, we have $\ell_z\le \ell_x\le r_z$. Thus $\ell_w<\ell_z\le r_z$ and $r_w\ge \ell_x\ge \ell_z$, implying $I_w\cap I_z\neq\emptyset$. Hence $w$ is adjacent to $z$ in $G-F'$. Therefore $p-w-z$ is a path of length two in $G-F'$, and so $\dist_{G-F'}(p,z)\le 2$. The case where $I_p$ lies to the right of $I_z$ is symmetric, using $x:=v_1$. Therefore every vertex of $C'_z$ is at distance at most two from $z$. Since $C'_z$ is already a $2$-club, it follows that $C_z=C'_z\cup\{z\}$ is a $2$-club in $G-F'$. All other components remain unchanged. Hence $F'$ is a feasible solution for $(G,k)$.

\medskip

\noindent \textbf{Case 2 (Rejection case).}
Neither $a$ nor $b$ satisfies the condition. In this case, the rule returns that $(G,k)$ is a no-instance. Let us prove the correctness of this case.

Assume that neither $a=v_{k+2}$ nor $b=v_{k+3}$ satisfies the test, i.e.,
there exist vertices $x_a,y_a\in N_G(a)$ and $x_b,y_b\in N_G(b)$ such that
\[
  N_G(x_a)\cap N_G(y_a)=\{a\}
  \qquad\text{and}\qquad
  N_G(x_b)\cap N_G(y_b)=\{b\}.
\]
We show that then $(G,k)$ is a no-instance of $2$-\textsc{Club Cluster Edge Deletion}.

\proofpara{Geometric placement of witnesses}
By the staircase property, every neighbor of $a$ intersects either all intervals
$I_{v_1},\ldots,I_{v_{k+1}}$ or all intervals
$I_{v_{k+3}},\ldots,I_{v_{2k+4}}$.
We say that a neighbor of $a$ intersects the \emph{left side} of $I_a$ if it
intersects $I_a$ in the segment $[\ell_a,\ell_b)$, and it intersects the
\emph{right side} of $I_a$ if it intersects $I_a$ in the segment
$(r_{v_{k+1}},r_a]$.
Analogously, a neighbor of $b$ intersects the \emph{left side} of $I_b$ if it
intersects $I_b$ in the segment $[\ell_b,\ell_{v_{k+4}})$, and it intersects
the \emph{right side} of $I_b$ if it intersects $I_b$ in the segment
$(r_a,r_b]$.

Since $x_a$ and $y_a$ have no common neighbor besides $a$, they cannot both intersect the same side of $I_a$. Indeed, suppose first that both intersect the left side of $I_a$. Then both $I_{x_a}$ and $I_{y_a}$ intersect the segment $[\ell_a,\ell_b)$. The interval $I_{v_{k+1}}$ contains this segment, since $\ell_{v_{k+1}}\le \ell_a$ and $r_{v_{k+1}}\ge \ell_b$. Hence $v_{k+1}$ is adjacent to both $x_a$ and $y_a$, contradicting $N_G(x_a)\cap N_G(y_a)=\{a\}$. 

Similarly, if both $x_a$ and $y_a$ intersected the right side of $I_a$, then both intervals would intersect the segment $(r_{v_{k+1}},r_a]$. 
Since \(b=v_{k+3}\), we have \(\ell_b\le r_{v_{k+1}}\) and \(r_a\le r_b\); therefore the segment \((r_{v_{k+1}},r_a]\) is contained in \(I_b\).
Therefore $b$ would be adjacent to both $x_a$ and $y_a$, again contradicting $N_G(x_a)\cap N_G(y_a)=\{a\}$. 

Therefore, after renaming $x_a$ and $y_a$ if necessary, we may assume that $x_a$ intersects the left side of $I_a$ and $y_a$ intersects the right side of $I_a$. 

The same argument applied to the witness pair $(x_b,y_b)$ shows that they cannot both intersect the same side of $I_b$: if both intersect the left side of $I_b$, then $a$ is a common neighbor of $x_b$ and $y_b$, while if both intersect the right side of $I_b$, then $v_{k+4}$ is a common neighbor of $x_b$ and $y_b$. Hence, after renaming if necessary, we may assume that $x_b$ intersects the left side of $I_b$ and $y_b$ intersects the right side of $I_b$.

We next record two consequences of the choice of the witnesses. Since $y_a$ intersects the right side of $I_a$, it is adjacent to $b$. If $x_a$ were also adjacent to $b$, then $b\in N_G(x_a)\cap N_G(y_a)$, contradicting $N_G(x_a)\cap N_G(y_a)=\{a\}$. Thus $x_a$ is not adjacent to $b$, and hence $r_{x_a}<\ell_b$. Similarly, since $x_b$ intersects the left side of $I_b$, it is adjacent to $a$. If $y_b$ were also adjacent to $a$, then $a\in N_G(x_b)\cap N_G(y_b)$, contradicting $N_G(x_b)\cap N_G(y_b)=\{b\}$. Thus $y_b$ is not adjacent to $a$, and hence $\ell_{y_b}>r_a$.

\proofpara{A cross pair is at distance at least $3$}
Consider the pair $(x_a,y_b)$.
By the observation above, $r_{x_{a}} < \ell_{b} \leq r_{a} <\ell_{y_{b}}$.
Therefore $I_{x_a}\cap I_{y_b}=\emptyset$, and consequently $x_a$ and $y_b$ are not adjacent.

Moreover, we claim that $x_a$ and $y_b$ have no common neighbor in $G$.
Suppose for contradiction that there exists a vertex $t \in V(G)$ adjacent
to both $x_a$ and $y_b$. Then $I_t$ intersects both $I_{x_a}$ and $I_{y_b}$.

Now consider the witness pair $(x_a,y_a)$.
Since $N_G(x_a)\cap N_G(y_a)=\{a\}$, the vertices $x_a$ and $y_a$ are not adjacent. 
Otherwise, $x_a$ and $y_a$ would be common neighbors of each other in addition
to $a$.
Hence the intervals $I_{x_a}$ and $I_{y_a}$ are disjoint.
Since \(I_{x_a}\) and \(I_{y_a}\) are disjoint, and since \(x_a\) intersects the left side of \(I_a\) whereas \(y_a\) intersects the right side of \(I_a\), we obtain
$r_{x_a}<\ell_{y_a}\le r_a < \ell_{y_b}$.
Since $I_t$ intersects both $I_{x_a}$ and $I_{y_b}$,
we have
$\ell_t \le r_{x_a}$ and $r_t \ge \ell_{y_b}$.
Hence the interval $I_t$ contains the entire segment
$[r_{x_a},\ell_{y_b}]$.
As $r_{x_a}<\ell_{y_a}<r_{y_a}<\ell_{y_b}$, it follows that $I_t$ must also intersect $I_{y_a}$.
Hence $t$ is adjacent to both $x_a$ and $y_a$.
This contradicts the assumption that
$N_G(x_a)\cap N_G(y_a)=\{a\}$.
Therefore no such vertex $t$ exists, and hence
$x_a$ and $y_b$ have no common neighbor.
Since they are not adjacent, it follows that $\dist_G(x_a,y_b)\ge 3$.

\proofpara{They cannot be separated within budget}
Let $F \subseteq E(G)$ with $|F| \le k$ be arbitrary, and consider the graph $G-F$.
First, all vertices of the clique $C=\{v_1,\ldots,v_{2k+4}\}$ lie in a single
connected component of $G-F$. 

We now show that $x_a$ and $y_b$ also lie in this component.
Since \(x_a\) intersects the left side of \(I_a\), the staircase structure implies that \(x_a\) is adjacent to each of \(v_1,\ldots,v_{k+1}\).
Hence $x_a$ has at least $k+1$ neighbors in $C$.
If $x_a$ were separated from $C$ in $G-F$, then all these edges would have to be
deleted, implying $|F|\ge k+1$, a contradiction. Thus $x_a$ lies in the same
component as $C$. The same argument shows that $y_b$, which is adjacent to all
vertices $v_{k+4},\ldots,v_{2k+4}$, also lies in this component.

Consequently, $x_a$ and $y_b$ lie in the same connected component of $G-F$.
From the previous argument, we have $\dist_G(x_a,y_b)\ge 3$, and hence
$\dist_{G-F}(x_a,y_b)\ge 3$. Therefore, this component is not a $2$-club.
Since $F$ was arbitrary, no solution of size at most $k$ yields a cluster of
$2$-clubs, and thus $(G,k)$ is a no-instance.
\end{proof}

\medskip
\noindent \textbf{Stage 5: Termination and kernel size.}
We now conclude Stage~4 and combine the previous reductions to derive
the final kernel size.

\begin{theorem}\label{thm:kernel-interval-s2}
$2$-\textsc{Club Cluster Edge Deletion} on interval graphs admits a
polynomial-time vertex kernel with at most $\mathcal{O}(k^{5})$ vertices.
\end{theorem}

\begin{proof}
We apply the reduction rules exhaustively.
First, apply Reduction Rule~\ref{rr:small-large-components}.
Consequently, every connected component $H$ of the resulting graph
satisfies $2 < \diam(H) < 3(k+1)$. 
Next, apply Reduction Rule~\ref{rr:cc-count}, which ensures that the
resulting graph has at most $k$ connected components.

Let $H$ be any connected component and let $n := |V(H)|$.
It is a standard property of interval graphs that
$\max\{\omega(H), \alpha(H)\} \ge \left\lceil \sqrt{n} \right\rceil$,
and such a clique or independent set can be found in polynomial time.
Let $T := (6k+7)^2$.
If $n > T^2$, then
$\left\lceil \sqrt{n} \right\rceil \ge T$,
and hence $H$ contains either:

\begin{itemize}
    \item an independent set of size at least $6k+7$, in which case Reduction Rule~\ref{rr:IS-neighbour-irrelevant} applies; or
    \item a clique of size at least $T$, in which case, by Lemma~\ref{lem:ES-clique-pattern-general}, there exists a clique of size $2k+4$ whose intervals form either a nested clique or a staircase clique. We then apply Reduction Rule~\ref{rr:nested-reduction} or Reduction Rule~\ref{rr:staircase-pair-test}, respectively.
\end{itemize}

Thus, as long as $|V(H)| > T^2$, some reduction rule is applicable,
either reducing the instance size or correctly rejecting it.
Therefore, when no rule applies, every connected component satisfies
$|V(H)| \le T^2 = \mathcal{O}(k^{4})$.
Since there are at most $k$ connected components, the total number of
vertices is at most
$k \cdot T^2 = \mathcal{O}(k^{5})$.
All reduction rules are safe and can be applied in polynomial time,
hence the resulting instance is a polynomial kernel equivalent to the
input.
\end{proof}

Having established hardness on split graphs, we next identify a natural class of graphs—unit interval graphs—on which the problem becomes polynomial-time solvable.

\begin{theorem}\label{thm:unitinterval-poly}
For every fixed integer $s\ge 1$, the problem \textsc{$s$-Club Cluster Edge Deletion} can be solved in time $\mathcal{O}(n^{2})$ on unit interval graphs.
\end{theorem}

\paragraph{Acknowledgments.}
The author sincerely thanks the anonymous reviewers of an earlier version of this work for their detailed feedback and insightful suggestions. Their comments led to substantial improvements in the exposition and organization of the paper, particularly in the presentation of the interval-graph kernelization.

\bibliographystyle{unsrt}
\bibliography{bibliography}

\clearpage
\appendix

\input{appendix}

\end{document}

%% file: appendix.tex
\section{More preliminaries}

\paragraph{Parameterized complexity and kernelization.}

We study the problem from the perspective of parameterized complexity,
taking the solution size $k$ as the parameter.
A \emph{kernelization} is a polynomial-time algorithm that transforms an
instance $(G,k)$ into an equivalent instance $(G',k')$ such that
$|V(G')|$ is bounded by a function of $k$ and $k'\le k$.
If this bound is polynomial in $k$, we say that the problem admits a
\emph{polynomial kernel}.
In this work, all kernels are \emph{vertex kernels}, meaning that the size
of the reduced instance is measured by the number of vertices. We refer to \cite{marekcygan,Downey} for further details on parameterized complexity.

\section{Omitted Proofs}

\subsection{Proof of Theorem~\ref{split NP-hard}}
\begin{proof}
We give a polynomial-time reduction from the classical \textsc{Clique} problem, which asks whether a given graph $G$ contains a clique of size at least $k$, and which is NP-hard even when restricted to $r$-regular graphs~\cite{10.5555/574848}.

\proofpara{Reduction construction}
Let $I = (G, k)$ be an instance of \textsc{Clique}, where $G$ is an $r$-regular graph with vertex set $V(G)$ and edge set $E(G)$.
We construct an equivalent instance $I' = (G', k')$ of 2-Club Cluster Edge Deletion, where $k' := r(n-k)$ and $n = |V(G)|$.

\smallskip
\noindent
\emph{Step~1: Subdivision and clique formation.}
For each edge $uv\in E(G)$, introduce a new vertex $e_{uv}$
adjacent to both $u$ and $v$, and delete the edge $uv$.
Let $C:=\{e_{uv} : uv \in E(G)\}$ denote the set of all subdivision vertices.
Next, make $C$ a clique by adding all edges between every pair of vertices in $C$.

\smallskip
\noindent
\emph{Step~2: Adding a universal auxiliary clique.}
Introduce a new clique $C'$ of size $k'+1$ and make every vertex of $C$ adjacent to every vertex of $C'$.
Thus, $C\cup C'$ induces a clique in $G'$, while the original vertices $V(G)$ form an independent set.
The clique $C'$ acts as a universal connector: since every vertex of $C$ is adjacent to every vertex of $C'$, the set $C \cup C'$ forms a clique, ensuring distance at most $2$ within this set and forcing it to remain in a single component in any feasible solution.

\smallskip
\noindent
\emph{Step~3: Parameter setting.}
This completes the construction of $G'$.
Clearly, $G'$ is a split graph, as its vertex set can be partitioned into the clique $C\cup C'$ and the independent set $V(G)$.
The construction can be performed in polynomial time.

\proofpara{(\textrm{If}) direction}
Assume that $G$ contains a clique $R\subseteq V(G)$ of size $k$.
We construct an edge-deletion set $F\subseteq E(G')$ as follows:
for each vertex $v\in V(G)\setminus R$, delete all edges incident to $v$ in $G'$.
Since each such vertex $v$ has degree $r$ (adjacent to $r$ subdivision vertices in $C$),
the number of deleted edges is exactly
$|F| = r(n-k) = k'$.
We now show that every connected component of $G'-F$ has diameter at most~2.
Each vertex $v\in V(G)\setminus R$ becomes an isolated vertex in $G'-F$, and hence forms a connected component of diameter~0.
All other vertices lie in a single connected component
$S := R \cup C \cup C'$.
Observe that:
\begin{itemize}
    \item any two vertices of $C\cup C'$ are adjacent (distance~1);
    \item for any distinct $x,y\in R$, there exists a subdivision vertex $e_{xy}\in C$ adjacent to both, since $xy \in E(G)$;
    hence, $\dist_S(x,y)=2$;
    \item for $x\in R$ and $y\in C\cup C'$, either $x$ is adjacent to $y$ (if $y=e_{x,z}$) or there exists $y'\in C$ such that $x$--$y'$--$y$ is a path of length~2.
\end{itemize}
Therefore, $\diam(S)\le 2$, and thus all connected components of $G'-F$ have diameter at most~2.
Hence, $I'$ is a yes-instance.

\proofpara{(\textrm{Only if}) direction}
Conversely, suppose that $I'=(G',k')$ is a yes-instance; that is,
there exists $F\subseteq E(G')$ with $|F|\le k'$ such that every connected component of $G'-F$ has diameter at most~2.
First observe that all vertices in $C\cup C'$ must lie in the same connected component,
since between any two such vertices there exist at least $k'+1$ internally edge-disjoint paths,
and deleting at most $k'$ edges cannot disconnect them.
Let this component be $S$.
Let $R := S\cap V(G)$ denote the subset of original vertices lying in $S$.
We claim that $R$ induces a clique of size at least~$k$ in $G$.
Indeed, if two distinct vertices $x,y\in R$ were \emph{not} adjacent in $G$,
then in $G'$ their shortest connecting path would have length~3
($x$–$e_{xz}$–$z$–$e_{zy}$–$y$ for some common neighbor $z$),
contradicting the assumption that every component has diameter at most~2.
Hence, $R$ induces a clique in $G$.
Finally, if $|R|\le k-1$, then at least $n-(k-1)$ vertices of $V(G)\setminus R$ must be isolated,
requiring the deletion of all their incident $r$ edges,
for a total of at least $r(n-(k-1)) = k'+r > k'$ edge deletions,
contradicting $|F|\le k'$.
Thus $|R|\ge k$, and $G$ contains a clique of size at least~$k$.

\proofpara{Conclusion}
We have shown that $I$ is a yes-instance of {\sc Clique}
if and only if $I'$ is a yes-instance of {\sc $2$-Club Cluster Edge Deletion}.
The construction can be carried out in polynomial time and produces a split graph.
Therefore, {\sc $2$-Club Cluster Edge Deletion} is NP-hard on split graphs.
\end{proof}

\subsection{Proof of Lemma~\ref{rr:small-large-components}}
\begin{proof}
For the first part, if $\diam(H)\le s$, then $H$ already forms a valid $s$-club and is disconnected from the rest of the graph, so removing $H$ does not affect the existence of any solution of size at most $k$.
For the second part, suppose $H$ has a shortest path between a pair of vertices
$P=(v_1,\ldots,v_{(s+1)(k+1)+1})$ of length $(s+1)(k+1)$. 
After deleting at most $k$ edges, there exist indices $i<j$ such that $v_i$ and $v_j$ lie in 
the same connected component and any $v_i$--$v_j$ path has length at least 
$s+1$. 
Hence $\mathrm{dist}_{G-F}(v_i, v_j) > s$, and no edge set of size at most $k$ can transform $G$ into an $s$-club cluster.
\end{proof}

\subsection{Proof of Lemma~\ref{rr:cc-count}}
\begin{proof}
By assumption, each remaining component $H$ has $\diam(H)>s$.
In any solution, $H$ must be split into smaller components of diameter at most $s$,
and this can only be achieved by deleting at least one edge inside $H$.
Thus each component requires at least one edge deletion.
If there are more than $k$ components, this would require more than $k$ deletions,
contradicting the budget.
\end{proof}

\subsection{Proof of Lemma~\ref{lem:IS-long-consecutive-neighbour}}
\begin{proof}
We first recall a standard property of connected interval graphs:
there exists a shortest path $P=(p_0,\ldots,p_d)$ with $d=\diam(H)$
that is \emph{dominating}, that is, every vertex of $H$ is adjacent to at
least one vertex of $P$.
Indeed, take vertices whose intervals have the leftmost left endpoint
and the rightmost right endpoint, respectively, and consider a shortest
path between them; any other interval intersecting the union of these
intervals must intersect at least one interval corresponding to a vertex
of the path.

Every vertex of $I$ lies in $H$ and hence is adjacent to at least one
vertex of $P$.
Assign each $v\in I$ to one fixed neighbour $p(v)\in V(P)$.
By the pigeonhole principle, there exists a vertex $u\in V(P)$ such that
\[
  |\{v\in I : p(v)=u\}|
    \;\ge\;
  \frac{|I|}{|V(P)|}
    \;\ge\;
  \frac{(s+1)(k+1)(6k+7)}{(s+1)(k+1)}
    \;=\;
  6k+7.
\]
Thus $u$ has at least $6k+7$ neighbours in $I$.

In an interval graph, the neighbours of a fixed vertex among an
independent set form a consecutive block when the independent set is
ordered by left endpoints:
if $I_{v_i}$ and $I_{v_j}$ intersect $I_u$ with $i<j$, then for every
index $t$ with $i<t<j$, the interval $I_{v_t}$ lies between them in the
ordering, and since the intervals in $I$ are pairwise disjoint, it follows
that $I_{v_t}$ also intersects $I_u$.
Hence the neighbours of $u$ in $I$ are $v_a,\ldots,v_b$ for some
indices $a\le b$, and from the above we have $b-a+1\ge 6k+7$.
In particular, $u$ is adjacent to a block of $6k+7$ consecutive vertices
of $I$, say $v_a,\ldots,v_{a+6k+6}$.
\end{proof}

\subsection{Proof of Observation~\ref{obs:interval-intersection-cross-edges}}

\begin{proof} Since $C_i$ is connected in $G-F$, for any two vertices of $C_i$ there is a path between them using only edges of $G-F$. Consecutive vertices on such a path have intersecting intervals. Hence the union of all intervals corresponding to vertices of $C_i$ is connected, and therefore $J(C_i)$ is an interval. Let $v\in X$. Then $I_v$ intersects both $J(C_1)$ and $J(C_2)$. Thus there exist vertices $w_1\in C_1$ and $w_2\in C_2$ such that $I_v\cap I_{w_1}\neq\emptyset$ and $I_v\cap I_{w_2}\neq\emptyset$. Hence $vw_1, vw_2\in E(G)$, unless $v=w_1$ or $v=w_2$. Since $C_1$ and $C_2$ are distinct connected components of $G-F$, the vertex $v$ cannot belong to both. If $v\notin C_1$, then $vw_1$ is an edge of $G$ between two distinct components of $G-F$, and so $vw_1\in F$. Otherwise $v\in C_1$, and then $v\notin C_2$, so $vw_2$ is an edge of $G$ between two distinct components of $G-F$, and hence $vw_2\in F$. Thus every vertex $v\in X$ is incident to at least one edge of $F$ whose endpoints lie in two distinct connected components of $G-F$. \end{proof}

\subsection{Proof of Claim~\ref{claim:middle-interval-squeezed}}
\begin{proof}
We first prove the bound on the left endpoint. Suppose, for a contradiction,
that
$\min_{y \in C_x} \ell_y \;\le\; r_{v_{a+k+1}}$.
Since $I_x \cap I_v \neq \emptyset$ and $x \in C_x$, it follows that
$J(C_x)$ intersects $I_v$, and hence the interval $J(C_u) \cap J(C_x)$
is non-empty.
By the assumption on the left endpoint, the interval
$J(C_u) \cap J(C_x)$ intersects the segment
$[\,\ell_{v_{a+k+1}},\, r_{v_{a+2k+2}}\,]$.
Recall that the vertices
\[
v_{a+k+1}, v_{a+k+2}, \ldots, v_{a+2k+2}
\]
belong to the independent set and their intervals are pairwise disjoint and
ordered by increasing left endpoints. Therefore, each of the intervals
\[
I_{v_{a+k+1}}, I_{v_{a+k+2}}, \ldots, I_{v_{a+2k+2}}
\]
intersects $J(C_u) \cap J(C_x)$.
Let $X := \{\, v_i : I_{v_i} \cap (J(C_u) \cap J(C_x)) \neq \emptyset \,\}$.
Then $|X| \;\ge\; k+1$.
By Observation~\ref{obs:interval-intersection-cross-edges}, every vertex
$v \in X$ is incident to at least one crossing edge, that is, an edge with one
endpoint in $C_u$ and the other in $C_x$. Since the vertices in $X$ form an
independent set, these edges are pairwise distinct. Consequently,
$|F| \;\ge\; |X| \;\ge\; k+1$,
which contradicts the assumption $|F| \le k$.
Therefore,
$\min_{y \in C_x} \ell_y \;>\; r_{v_{a+k+1}}$.

The proof for the right endpoint is symmetric. For the sake of contradiction assume that $\max_{y \in C_x} r_y \;\ge\; \ell_{v_{a+3k+5}}$.
Then, by considering the intervals
$I_{v_{a+2k+2}}, I_{v_{a+2k+3}}, \ldots, I_{v_{a+3k+5}}$,
and applying the same argument, we again obtain at least $k+1$ distinct
crossing edges, contradicting $|F| \le k$. Hence,
$\max_{y \in C_x} r_y \;<\; \ell_{v_{a+3k+5}}$.
Combining the two bounds, we conclude that
\[
J(C_x) \subset \bigl(r_{v_{a+k+1}},\, \ell_{v_{a+3k+5}}\bigr).
\]
This finishes the proof of Claim~\ref{claim:middle-interval-squeezed}.
\end{proof}

\subsection{Proof of Claim~\ref{claim:merge-clubs}}
\begin{proof}
We may assume that $F$ is inclusion-minimal among the feasible edge-deletion sets under consideration. Hence no edge of $F$ has both endpoints in the same connected component of $G-F$: adding such an edge back would not merge two components and could only decrease distances inside the component.
Let $\widehat{G}$ be the graph obtained from $G-F$ by restoring all edges
between $C_u$ and $C_x$. Since $C_u$ and $C_x$ are already $s$-clubs in
$G-F$, it is enough to show that
$\dist_{\widehat{G}}(p,q)\le s$
for every $p\in C_u$ and every $q\in C_x$.
By Claim~\ref{claim:middle-interval-squeezed}, for every $q\in C_x$ we have
$I_q \subset \bigl(r_{v_{a+k+1}},\,\ell_{v_{a+3k+5}}\bigr)$.
Since \(u\) is adjacent to every vertex \(v_a,\ldots,v_{a+6k+6}\), and these intervals are pairwise disjoint and ordered by left endpoints, the interval \(I_u\) spans the region between \(I_{v_{a+k+1}}\) and \(I_{v_{a+3k+5}}\). Hence \[ (r_{v_{a+k+1}},\ell_{v_{a+3k+5}})\subseteq I_u. \]
Hence every vertex of $C_x$ is adjacent to
$u$ in $\widehat{G}$.
Therefore, if $p=u$, then $\dist_{\widehat{G}}(p,q)=1$ for every
$q\in C_x$. Hence, in the rest of the proof, we may assume $p\neq u$.

It remains to consider the case where $\dist_{G-F}(p,u)=s$. Since $p$ and $u$ lie in the same connected component $C_u$ of $G-F$, and we may assume $F$ is inclusion-minimal, no edge of $F$ has both endpoints in $C_u$. Hence, if $p$ were adjacent to $u$ in $G$, then the edge $pu$ would be present in $G-F$, contradicting $\dist_{G-F}(p,u)=s\ge 2$. Therefore $p$ is not adjacent to $u$ in $G$, and so $I_p\cap I_u=\emptyset$. Thus $I_p$ lies completely either to the left or to the right of $I_u$. Without loss of generality, assume that $I_p$ lies completely to the right of $I_u$; the case where $I_p$ lies completely to the left of $I_u$ is symmetric.
We first show that at least one vertex among
\[
v_a,v_{a+1},\ldots,v_{a+k}
\]
belongs to $C_u$. Suppose this is not the case. Then each of these $k+1$ vertices belongs
to a component of $G-F$ distinct from $C_u$. Since all of them are adjacent
to $u$, each of them is incident to at least one edge of $F$ with one endpoint
in $C_u$. These $k+1$ edges are distinct, contradicting $|F|\le k$.
Let $v_i$ be a vertex in
\[
\{v_a,v_{a+1},\ldots,v_{a+k}\}\cap C_u
\]
with minimum index $i$.
In the symmetric case, one uses a vertex of \[\{v_{a+5k+6},\ldots,v_{a+6k+6}\}\cap C_u,\] whose existence follows by the same \(k+1\)-edge deletion argument.
Since $C_u$ is an $s$-club in $G-F$, we have
$\dist_{G-F}(p,v_i) \le s$. Let \(P=(p=p_0,p_1,\ldots,p_\ell=v_i)\) be a shortest \(p\)-\(v_i\) path in \(G-F\). Since \(C_u\) is an \(s\)-club, we have \(\ell\le s\). 
By the path-interval union observation, the union  $\bigcup_{j=0}^{\ell} I_{p_j}$ is a connected interval containing both \(I_p\) and \(I_{v_i}\). Now let \(q\in C_x\).
By Claim~\ref{claim:middle-interval-squeezed}, \(I_q\) lies in the interval \((r_{v_{a+k+1}},\ell_{v_{a+3k+5}})\).
Since \(i\in\{a,\ldots,a+k\}\), the interval \(I_q\) lies strictly to the right of \(I_{v_i}\), and in particular \(I_q\cap I_{v_i}=\emptyset\). 
Since $I_p$ lies to the right of $I_u$, while $I_{v_i}$ lies to the left of the interval containing $I_q$, the interval $I_q$ is contained in the real-line region between $I_{v_i}$ and $I_p$. The union of the intervals along $P$ is a connected interval containing both $I_{v_i}$ and $I_p$; therefore this union must intersect $I_q$. Hence $I_q$ intersects $I_{p_j}$ for some $j\in\{0,\ldots,\ell\}$. Moreover, $j<\ell$, because $p_\ell=v_i$ and $I_q\cap I_{v_i}=\emptyset$.
The vertex \(p_j\) lies in \(C_u\), while \(q\) lies in \(C_x\). In the graph obtained from \(G-F\) by restoring all edges between \(C_u\) and \(C_x\), the edge \(p_jq\) is present. Therefore \[ \dist_{\widehat{G}}(p,q)\le j+1\le \ell\le s. \]
Therefore, the component $C_u \cup C_x$ has diameter at most $s$ in
$\widehat{G}$, and hence is an $s$-club. This completes the proof of
Claim~\ref{claim:merge-clubs}.
\end{proof}

\subsection{Proof of Lemma~\ref{rr:nested-reduction}}

\begin{proof} Let \(u:=v_{k+2}\) and let \(G':=G-u\).

\proofpara{Preliminary observations} First observe that the vertices \(v_1,\ldots,v_{2k+3}\) of the nested clique must stay in one connected component after deleting at most \(k\) edges. Indeed, separating a clique of size \(2k+3\) into two nonempty parts requires deleting at least \(2k+2>k\) edges. 

We also record two simple domination properties. Since \[ I_{v_1}\supseteq I_{v_2}\supseteq \cdots \supseteq I_{v_{2k+3}}, \] we have \(I_{v_i}\supseteq I_u\) for every \(i\le k+1\). Hence every neighbor of \(u\) is adjacent to each of \(v_1,\ldots,v_{k+1}\). Indeed, if \(w\in N_G(u)\), then \(I_w\cap I_u\neq\emptyset\), and this intersection point also lies in each \(I_{v_i}\) with \(i\le k+1\). 

Finally, for every \(u'\in\{v_{k+3},\ldots,v_{2k+3}\}\), we have \(I_{u'}\subseteq I_u\). Therefore every neighbor of \(u'\) is also a neighbor of \(u\), that is, \[ N_G[u']\subseteq N_G[u]. \] 

\medskip \noindent\emph{Forward direction.} Assume that \((G,k)\) is a yes-instance. Let \(F\subseteq E(G)\) be an inclusion-minimal solution with \(|F|\le k\) such that \(G-F\) is a cluster of \(s\)-clubs. Thus no edge of \(F\) has both endpoints in the same connected component of \(G-F\), since adding such an edge back would not merge two components and could only decrease distances inside that component. 

Let \(C\) be the connected component of \(G-F\) containing the nested clique \(v_1,\ldots,v_{2k+3}\). By the observation above, all clique vertices lie in \(C\). 

We first note that every neighbor of \(u\) belongs to \(C\). Indeed, every neighbor of \(u\) is adjacent to all vertices \(v_1,\ldots,v_{k+1}\), which all belong to \(C\). If such a neighbor were outside \(C\), then the \(k+1\) edges from it to \(v_1,\ldots,v_{k+1}\) would all have to be deleted, contradicting \(|F|\le k\). 

Let \(F':=F\cap E(G')\). All components of \(G-F\) different from \(C\) are unchanged in \(G'-F'\). Thus it suffices to show that \(C\setminus\{u\}\) is an \(s\)-club in \(G'-F'\). 

Take any two vertices \(p,q\in C\setminus\{u\}\). Since \(C\) is an \(s\)-club in \(G-F\), there is a \(p\)-\(q\) path of length at most \(s\) in \(G-F\). Choose such a shortest path \(P\). If \(P\) avoids \(u\), then \(P\) is also present in \(G'-F'\), and we are done. 

Otherwise, write \[ P=(p,\ldots,x,u,y,\ldots,q), \] where \(x\) and \(y\) are the predecessor and successor of \(u\) on \(P\). Since \(x,y\in N_G(u)\), and every neighbor of \(u\) is adjacent to each of \(v_1,\ldots,v_{k+1}\), choose any \(r\in\{v_1,\ldots,v_{k+1}\}\). Then \(xr,yr\in E(G)\). Moreover, \(x,y,r\in C\). By the inclusion-minimality of \(F\), the edges \(xr\) and \(yr\) are not deleted. Hence replacing the subpath \(x-u-y\) by \(x-r-y\) gives a \(p\)-\(q\) walk in \(G'-F'\) of length no larger than \(|P|\). Therefore \[ \dist_{G'-F'}(p,q)\le s. \] Thus \(C\setminus\{u\}\) remains an \(s\)-club, and so \(F'\) is a feasible solution for \((G',k)\). Hence \((G',k)\) is a yes-instance.

\medskip \noindent\emph{Backward direction.} Assume that \((G',k)\) is a yes-instance. Let \(F'\subseteq E(G')\) with \(|F'|\le k\) be such that \(G'-F'\) is a cluster of \(s\)-clubs. 

Since \(v_1,\ldots,v_{2k+3}\setminus\{u\}\) form a clique of size \(2k+2\), they remain in one connected component of \(G'-F'\); separating them would require deleting at least \(2k+1>k\) edges. Let \(C'\) denote this component. 

We first show that every neighbor of \(u\) belongs to \(C'\). Let \(w\in N_G(u)\). Then \(w\) is adjacent to all vertices \(v_1,\ldots,v_{k+1}\), and these vertices belong to \(C'\). If \(w\notin C'\), then all \(k+1\) edges from \(w\) to \(v_1,\ldots,v_{k+1}\) would have to be deleted by \(F'\), contradicting \(|F'|\le k\). Thus \(N_G(u)\subseteq C'\).

Now view \(F'\) as an edge set of \(G\). Since all neighbors of \(u\) already lie in \(C'\), adding \(u\) back does not merge \(C'\) with any other component. All components different from \(C'\) remain unchanged \(s\)-clubs. It remains to show that \[ C:=C'\cup\{u\} \] is an \(s\)-club in \(G-F'\).

Distances inside \(C'\) do not increase when \(u\) is added. Hence it suffices to show that every vertex \(a\in C'\) satisfies $\dist_{G-F'}(a,u)\le s$. Choose any vertex $u'\in \{v_{k+3},\ldots,v_{2k+3}\}\subseteq C'$.  As observed above, \(N_G[u']\subseteq N_G[u]\). Since \(C'\) is an \(s\)-club in \(G'-F'\), there is an \(a\)-\(u'\) path of length at most \(s\) in \(G'-F'\). Let \[ P=(a=w_0,w_1,\ldots,w_{\ell-1},w_\ell=u') \] be a shortest such path, so \(\ell\le s\).

If \(a=u'\), then \(a\) is adjacent to \(u\), and we are done. Otherwise, \(w_{\ell-1}\) is adjacent to \(u'\) in \(G'-F'\). Since \(N_G[u']\subseteq N_G[u]\), the vertex \(w_{\ell-1}\) is adjacent to \(u\) in \(G\). Moreover, no edge incident with \(u\) belongs to \(F'\), because \(u\notin V(G')\). Therefore the edge \(w_{\ell-1}u\) is present in \(G-F'\). Replacing the last vertex \(u'\) of \(P\) by \(u\) gives a walk  $(a=w_0,w_1,\ldots,w_{\ell-1},u)$  of length \(\ell\le s\) in \(G-F'\). Thus  $\dist_{G-F'}(a,u)\le s$. 
Since \(a\in C'\) was arbitrary, every vertex of \(C'\) is at distance at most \(s\) from \(u\). Hence \(C'\cup\{u\}\) is an \(s\)-club. Therefore \(F'\) is a feasible solution for \((G,k)\), and so \((G,k)\) is a yes-instance.

This proves both directions and completes the proof of Lemma~\ref{rr:nested-reduction}. 
\end{proof}

\subsection{Proof of Lemma~\ref{lem:IS-neighbour-irrelevant}} \begin{proof} \noindent\emph{Forward direction.} Assume that $(G,k)$ is a yes-instance. Let $F\subseteq E(G)$ be an inclusion-minimal solution with $|F|\le k$ such that $G-F$ is a cluster of $s$-clubs. Thus no edge of $F$ has both endpoints in the same connected component of $G-F$, since adding such an edge back would not merge two components and could only decrease distances inside that component. Let $v:=v_{a+2k+2}$ be the vertex deleted by the rule, and let $C_u$ and $C_v$ denote the connected components of $G-F$ containing $u$ and $v$, respectively. 
We distinguish two cases. 

\medskip \noindent\textbf{Case 1:} $C_u=C_v$. 

Let $C:=C_u=C_v$, define $G':=G-v$, and let $F':=F\cap E(G')$. We show that $C\setminus\{v\}$ remains an $s$-club in $G'-F'$. Let $p,q\in C\setminus\{v\}$. Since $C$ is an $s$-club in $G-F$, there is a $p$--$q$ path of length at most $s$ in $G-F$. If such a shortest path avoids $v$, then it is preserved in $G'-F'$. Otherwise, let \[ P=(p,\ldots,x,v,y,\ldots,q) \] be a shortest $p$--$q$ path in $G-F$ using $v$. Since $I_u\supseteq I_v$, every neighbor of $v$ in $G$ is also adjacent to $u$ in $G$. Hence $xu,yu\in E(G)$. Moreover, $x,u,y$ all belong to the same component $C$ of $G-F$. By the inclusion-minimality of $F$, the edges $xu$ and $yu$ are not deleted. Therefore both edges are present in $G-F$ and hence also in $G'-F'$. Replacing the subpath $x-v-y$ of $P$ by $x-u-y$ gives a $p$--$q$ walk of length no larger than $|P|$ in $G'-F'$. Hence \[ \dist_{G'-F'}(p,q)\le s. \] Thus $C\setminus\{v\}$ is an $s$-club in $G'-F'$, and all other components are unchanged. Therefore $(G',k)$ is a yes-instance.

\medskip \noindent\textbf{Case 2:} $C_u\neq C_v$. 

Since $I_v\cap I_v\neq\emptyset$, we may apply Claim~\ref{claim:merge-clubs} with $x=v$. Hence, if we restore all edges between $C_u$ and $C_v$, the union $C_u\cup C_v$ becomes an $s$-club. Let $F^\ast$ be the resulting edge-deletion set. Clearly $|F^\ast|\le |F|\le k$, and $G-F^\ast$ is still a cluster of $s$-clubs. In the graph $G-F^\ast$, the vertices $u$ and $v$ belong to the same $s$-club component. The replacement argument from Case~1 applies to $F^\ast$: any edge needed inside one of the original components is present by the inclusion-minimality of $F$, and any edge needed between $C_u$ and $C_v$ has been restored. Hence deleting $v$ yields a feasible solution for $(G',k)$. Therefore $(G',k)$ is a yes-instance.

\medskip \noindent\emph{Backward direction.} Let $G':=G-v$ and suppose that $(G',k)$ is a yes-instance. Choose an inclusion-minimal solution $F'\subseteq E(G')$ with $|F'|\le k$ such that $G'-F'$ is a cluster of $s$-clubs. Let $\mathcal C'$ be the set of connected components of $G'-F'$, and let $C_u$ be the component containing $u$. 

We first enlarge the component containing $u$. Let $\mathcal D$ be the set of components $D\in\mathcal C'$ with $D\neq C_u$ such that there exists $x\in D$ with $I_x\cap I_v\neq\emptyset$. Starting from $F'$, restore all edges between $C_u$ and every component $D\in\mathcal D$, and let the resulting edge-deletion set be $\widetilde F'$. Clearly $|\widetilde F'|\le |F'|\le k$.

The proofs of Claims~\ref{claim:middle-interval-squeezed} and \ref{claim:merge-clubs} use only the inherited interval representation and the fact that the considered sets are components of an edge-deleted interval graph; therefore they apply verbatim to the induced interval graph \(G'=G-v\), with \(I_v\) used as the reference interval.

For every component $D\in\mathcal D$, Claim~\ref{claim:merge-clubs}, applied to the original components $C_u$ and $D$, shows that $C_u\cup D$ is an $s$-club after restoring the edges between them. Moreover, for every such component $D$, Claim~\ref{claim:middle-interval-squeezed} implies that every vertex of $D$ is adjacent to $u$. Hence, after all restorations, any two vertices belonging to two different components of $\mathcal D$ are at distance at most two via $u$. Since $s\ge 2$, the resulting component containing $u$ is still an $s$-club. Let this component be denoted by $\widetilde C_u$. All other components remain unchanged $s$-clubs. 

Furthermore, after these restorations, all neighbors of $v$ lie in $\widetilde C_u$. Indeed, if $w\in N_G(v)$, then $I_w\cap I_v\neq\emptyset$, and therefore the component of $G'-F'$ containing $w$ is either $C_u$ itself or belongs to $\mathcal D$. 

We now view $\widetilde F'$ as an edge set of $G$. Since every neighbor of $v$ lies in $\widetilde C_u$, adding $v$ back only attaches $v$ to $\widetilde C_u$ and does not merge any other component. It remains to show that  $\widetilde C_u\cup\{v\}$  is an $s$-club in $G-\widetilde F'$. Let $p\in \widetilde C_u$. If \[ \dist_{G-\widetilde F'}(p,u)\le s-1, \] then the edge $uv$, which is not deleted by $\widetilde F'$, gives $\dist_{G-\widetilde F'}(p,v)\le s$. So assume that $\dist_{G-\widetilde F'}(p,u)=s$.

We claim that $p$ is not adjacent to $u$ in $G$. Indeed, if $p$ belonged to the original component $C_u$, then such an edge $pu$ would not be deleted by the inclusion-minimality of $F'$. If $p$ belonged to some component $D\in\mathcal D$, then the edge $pu$, if present in $G$, was restored when forming $\widetilde F'$. In either case, $pu$ would be present in $G-\widetilde F'$, contradicting $\dist_{G-\widetilde F'}(p,u)=s\ge 2$. Hence $p$ is not adjacent to $u$ in $G$, and therefore $I_p\cap I_u=\emptyset$. Thus $I_p$ lies completely either to the left or to the right of $I_u$. 

Assume without loss of generality that $I_p$ lies completely to the right of $I_u$; the other case is symmetric. We first show that some vertex among \[ v_a,v_{a+1},\ldots,v_{a+k} \] belongs to $\widetilde C_u$. Suppose not. Since all these $k+1$ vertices are adjacent to $u$, separating all of them from the component containing $u$ would require deleting at least $k+1$ distinct edges incident with $u$, contradicting $|\widetilde F'|\le k$. Hence there is a vertex  $v_i\in \{v_a,\ldots,v_{a+k}\}\cap \widetilde C_u$. 
Since $\widetilde C_u$ is an $s$-club in $G'-\widetilde F'$, there is a shortest $p$--$v_i$ path \[ P=(p=p_0,p_1,\ldots,p_\ell=v_i) \] in $G'-\widetilde F'$ with $\ell\le s$. By the path-interval union observation,  $\bigcup_{j=0}^{\ell} I_{p_j}$  is a connected interval containing both $I_p$ and $I_{v_i}$. 

The interval $I_v$ lies strictly to the right of $I_{v_i}$ and to the left of $I_p$. Moreover, since $v_i$ and $v$ are distinct vertices of the independent set, $I_v\cap I_{v_i}=\emptyset$. Therefore the connected union of intervals along $P$ intersects $I_v$ in some interval $I_{p_j}$ with $j<\ell$. Hence $v$ is adjacent to $p_j$ in $G$. Since $\widetilde F'\subseteq E(G')=E(G-v)$, no edge incident with $v$ is deleted by $\widetilde F'$. Thus the edge $vp_j$ is present in $G-\widetilde F'$, and consequently $ \dist_{G-\widetilde F'}(p,v)\le j+1\le \ell\le s$. 

The symmetric case, where $I_p$ lies completely to the left of $I_u$, is handled analogously using a vertex of  $\{v_{a+5k+6},\ldots,v_{a+6k+6}\}\cap \widetilde C_u$,  whose existence follows from the same $k+1$-edge deletion argument. 

Thus every vertex of $\widetilde C_u$ is at distance at most $s$ from $v$. Since $\widetilde C_u$ is already an $s$-club, the component $\widetilde C_u\cup\{v\}$ is an $s$-club. All other components remain unchanged. Therefore $\widetilde F'$ is a feasible solution for $(G,k)$, and hence $(G,k)$ is a yes-instance. 

This completes the proof of Lemma~\ref{thm:unitinterval-poly}. \end{proof}

\subsection{Proof of Lemma~\ref{lem:IS-neighbour-irrelevant}}
\begin{proof}
We first establish a structural lemma showing that no optimal solution needs to split a maximal clique.  Then we exploit the canonical clique path of a unit interval graph to obtain a dynamic‐programming algorithm.

\medskip
\noindent
\textbf{Lemma 1 (No Clique Splitting).}
Let $G$ be a unit interval graph with canonical clique path
$C_1,C_2,\dots,C_m$, where every vertex of $G$ appears in a contiguous subsequence of cliques and consecutive cliques intersect.
Let $F\subseteq E(G)$ be an edge-deletion set such that every connected component of $G-F$ has diameter at most~$s$.
Then there exists a deletion set $F^{\star}$ with $|F^{\star}|\le |F|$ such that
(i)~each maximal clique $C_i$ is wholly contained in one component of $G-F^{\star}$,
and (ii)~every component of $G-F^{\star}$ has diameter~$\le s$.

\begin{proof}
Let $F$ be a minimum-size feasible deletion set that minimizes
the number of split cliques.
Assume, for contradiction, that some $C_p$ is split by $F$:
there exist components $A$ and $B$ of $G-F$
with $A\cap C_p\neq\emptyset$ and $B\cap C_p\neq\emptyset$.
Thus $F$ contains all edges between $A\cap C_p$ and $B\cap C_p$.

Each component of $G-F$ spans at most $s$ consecutive cliques
of the canonical path.
Hence there exist index intervals $[L_A,R_A]$ and $[L_B,R_B]$
of length $\le s$ such that $A$ and $B$ are contained in
$\bigcup_{t\in[L_A,R_A]}C_t$ and $\bigcup_{t\in[L_B,R_B]}C_t$,
respectively, and both contain $C_p$.
Define two windows:
\[
W_L=[p-s+1,p]\cap[1,m]
\qquad\text{and}\qquad
W_R=[p,p+s-1]\cap[1,m].
\]
Construct $F^\star$ by assigning $C_p$ wholly to one side, say to the right:
\[
F^\star = \bigl(F\setminus \tbinom{C_p}{2}\bigr)
          \cup
          \bigl\{\, uv\in E(G)\mid u\in C_p,\ v\in C_t,\ t<p
                   \text{ or } t>p+s-1 \,\bigr\}.
\]
That is, we restore all internal edges of $C_p$ and cut all edges
from $C_p$ to cliques outside the right window.

\proofpara{Feasibility}
All vertices of $C_p$ now have neighbors only in
$\bigcup_{t\in W_R}C_t$, a union of at most $s$ consecutive cliques.
By the distance property of unit interval graphs,
this subgraph has diameter at most~$s$.
Hence every connected component of $G-F^\star$ has diameter~$\le s$.

\proofpara{Size}
$F^\star$ deletes no more edges than $F$:
we restore all edges internal to $C_p$
and add only edges from $C_p$ to outside the chosen window.
By symmetry, we can pick the cheaper of the two assignments
(to $W_L$ or $W_R$), ensuring $|F^\star|\le|F|$.

Finally, since all edges within $C_p$ are present in $G-F^\star$,
the clique $C_p$ is now contained entirely within one component,
and no new clique becomes split.
Thus $F^\star$ has strictly fewer split cliques than $F$,
contradicting the minimal choice of $F$.
Therefore no optimal solution splits a maximal clique.
\end{proof}

\medskip
\noindent
\textbf{Dynamic‐Programming Formulation.}
By Lemma 1, every component of an optimal solution consists of a contiguous block
of at most $s$ consecutive maximal cliques of the canonical clique path.
Hence an optimal solution corresponds exactly to a \emph{partition}
of the index set $\{1,\dots,p\}$ into consecutive blocks
$[b_1,e_1],[b_2,e_2],\dots,[b_r,e_r]$
such that $b_1=1$, $e_r=m$, and $b_{t+1}=e_t+1$.
For any boundary between two consecutive blocks, we must delete every edge having
endpoints in the two adjacent blocks.
Let $\mathrm{cost}(i,j)$ denote the number of edges having one endpoint in
$\bigcup_{t\le i}C_t$ and the other in $\bigcup_{t>i}C_t$ when we place a cut between
$C_i$ and $C_{i+1}$.
These costs can be computed in $O(n^2)$ preprocessing by scanning vertex intervals.
Define the DP table
$\mathrm{DP}[t]$= minimum number of deletions needed so that
$C_1,\dots,C_t$ are partitioned into valid $s$-club components.

The recurrence is
$\mathrm{DP}[t]=
\min_{1\le j\le s}
\bigl(\mathrm{DP}[t-j] + \mathrm{cost}(t-j,t-j+1)\bigr)$,
with base case $\mathrm{DP}[0]=0$ and $\mathrm{cost}(0,1)=0$.
Intuitively, we end the last component at $C_t$, whose length may be at most $s$,
and pay the deletion cost for the cut between $C_{t-j}$ and $C_{t-j+1}$.
Because there are $p=O(n)$ cliques in the canonical path and $s$ is fixed,
the DP runs in $O(sn)$ time once the cut costs are known.

Finally, reconstructing the partition from the DP yields the edge-deletion set
that achieves the minimum number of deletions.
Thus \textsc{$s$-Club Cluster Edge Deletion} is solvable in polynomial time on unit interval graphs.
\end{proof}

%% file: bibliography.bib
@Inbook{Erdos1987,
author="Erd{\"o}s, P.
and Szckeres, G.",
editor="Gessel, Ira
and Rota, Gian-Carlo",
title="A Combinatorial Problem in Geometry",
bookTitle="Classic Papers in Combinatorics",
year="1987",
publisher="Birkh{\"a}user Boston",
address="Boston, MA",
pages="49--56",
isbn="978-0-8176-4842-8",
doi="10.1007/978-0-8176-4842-8_3",
url="https://doi.org/10.1007/978-0-8176-4842-8_3"
}

@InProceedings{konstantinidis_et_al:LIPIcs.MFCS.2019.12,
  author =	{Konstantinidis, Athanasios L. and Papadopoulos, Charis},
  title =	{{Cluster Deletion on Interval Graphs and Split Related Graphs}},
  booktitle =	{44th International Symposium on Mathematical Foundations of Computer Science (MFCS 2019)},
  pages =	{12:1--12:14},
  series =	{Leibniz International Proceedings in Informatics (LIPIcs)},
  ISBN =	{978-3-95977-117-7},
  ISSN =	{1868-8969},
  year =	{2019},
  volume =	{138},
  editor =	{Rossmanith, Peter and Heggernes, Pinar and Katoen, Joost-Pieter},
  publisher =	{Schloss Dagstuhl -- Leibniz-Zentrum f{\"u}r Informatik},
  address =	{Dagstuhl, Germany},
  URL =		{https://drops.dagstuhl.de/entities/document/10.4230/LIPIcs.MFCS.2019.12},
  URN =		{urn:nbn:de:0030-drops-109568},
  doi =		{10.4230/LIPIcs.MFCS.2019.12}
}

@InProceedings{10.1007/978-3-642-17493-3_8,
author="Cao, Yixin
and Chen, Jianer",
editor="Raman, Venkatesh
and Saurabh, Saket",
title="Cluster Editing: Kernelization Based on Edge Cuts",
booktitle="Parameterized and Exact Computation",
year="2010",
publisher="Springer Berlin Heidelberg",
address="Berlin, Heidelberg",
pages="60--71",
isbn="978-3-642-17493-3"
}

@book{984029,
author = {Golumbic, Martin Charles},
title = {Algorithmic Graph Theory and Perfect Graphs (Annals of Discrete Mathematics, Vol 57)},
year = {2004},
isbn = {0444515305},
publisher = {North-Holland Publishing Co.},
address = {NLD}
}

@article{BONOMO2015600,
title = {A one-to-one correspondence between potential solutions of the cluster deletion problem and the minimum sum coloring problem, and its application to P4-sparse graphs},
journal = {Information Processing Letters},
volume = {115},
number = {6},
pages = {600-603},
year = {2015},
issn = {0020-0190},
doi = {https://doi.org/10.1016/j.ipl.2015.02.007},
url = {https://www.sciencedirect.com/science/article/pii/S0020019015000241},
author = {Flavia Bonomo and Guillermo Duran and Amedeo Napoli and Mario Valencia-Pabon}
}

@book{10.5555/984029,
author = {Golumbic, Martin Charles},
title = {Algorithmic Graph Theory and Perfect Graphs (Annals of Discrete Mathematics, Vol 57)},
year = {2004},
isbn = {0444515305},
publisher = {North-Holland Publishing Co.},
address = {NLD}
}

@inproceedings{Foldes1977SplitGraphs,
  author    = {F{\"o}ldes, S. and Hammer, P. L.},
  title     = {Split graphs},
  booktitle = {Proceedings of the 8th South-Eastern Conference on Combinatorics, Graph Theory and Computing},
  series    = {Congressus Numerantium},
  volume    = {19},
  pages     = {311--315},
  year      = {1977}
}

@Inbook{Berkhin2006,
author="Berkhin, P.",
editor="Kogan, Jacob
and Nicholas, Charles
and Teboulle, Marc",
title="A Survey of Clustering Data Mining Techniques",
bookTitle="Grouping Multidimensional Data: Recent Advances in Clustering",
year="2006",
publisher="Springer Berlin Heidelberg",
address="Berlin, Heidelberg",
pages="25--71",
isbn="978-3-540-28349-2",
doi="10.1007/3-540-28349-8_2",
url="https://doi.org/10.1007/3-540-28349-8_2"
}

@article{FORTUNATO201075,
title = {Community detection in graphs},
journal = {Physics Reports},
volume = {486},
number = {3},
pages = {75-174},
year = {2010},
issn = {0370-1573},
doi = {https://doi.org/10.1016/j.physrep.2009.11.002},
url = {https://www.sciencedirect.com/science/article/pii/S0370157309002841},
author = {Santo Fortunato}
}

@article{Alizadeh1995,
  author    = {F. Alizadeh and R. M. Karp and L. A. Newberg and D. K. Weisser},
  title     = {Physical mapping of chromosomes: A combinatorial problem in molecular biology},
  journal   = {Algorithmica},
  volume    = {13},
  number    = {1--2},
  pages     = {52--76},
  year      = {1995},
  doi       = {10.1007/BF01188581},
  publisher = {Springer}
}

@article{HARTUV2000249,
title = {An Algorithm for Clustering cDNA Fingerprints},
journal = {Genomics},
volume = {66},
number = {3},
pages = {249-256},
year = {2000},
issn = {0888-7543},
doi = {https://doi.org/10.1006/geno.2000.6187},
url = {https://www.sciencedirect.com/science/article/pii/S0888754300961871},
author = {Erez Hartuv and Armin O Schmitt and Jörg Lange and Sebastian Meier-Ewert and Hans Lehrach and Ron Shamir}
}

@article{
doi:10.1073/pnas.122653799,
author = {M. Girvan  and M. E. J. Newman },
title = {Community structure in social and biological networks},
journal = {Proceedings of the National Academy of Sciences},
volume = {99},
number = {12},
pages = {7821-7826},
year = {2002},
doi = {10.1073/pnas.122653799},
URL = {https://www.pnas.org/doi/abs/10.1073/pnas.122653799},
eprint = {https://www.pnas.org/doi/pdf/10.1073/pnas.122653799}}

@article{GAO20132763,
title = {The cluster deletion problem for cographs},
journal = {Discrete Mathematics},
volume = {313},
number = {23},
pages = {2763-2771},
year = {2013},
issn = {0012-365X},
doi = {https://doi.org/10.1016/j.disc.2013.08.017},
url = {https://www.sciencedirect.com/science/article/pii/S0012365X13003543},
author = {Yong Gao and Donovan R. Hare and James Nastos}
}

@article{doi:10.1142/S0129054107004656,
author = {DESSMARK, ANDERS and JANSSON, JESPER and LINGAS, ANDRZEJ and LUNDELL, EVA-MARTA and PERSSON, MIA},
title = {ON THE APPROXIMABILITY OF MAXIMUM AND MINIMUM EDGE CLIQUE PARTITION PROBLEMS},
journal = {International Journal of Foundations of Computer Science},
volume = {18},
number = {02},
pages = {217-226},
year = {2007},
doi = {10.1142/S0129054107004656},

URL = { 
    
        https://doi.org/10.1142/S0129054107004656
    
    

},
eprint = { 
    
        https://doi.org/10.1142/S0129054107004656
    
    

}
}

@book{10.5555/574848,
author = {Garey, Michael R. and Johnson, David S.},
title = {Computers and Intractability; A Guide to the Theory of NP-Completeness},
year = {1990},
isbn = {0716710455},
publisher = {W. H. Freeman \& Co.},
address = {USA}
}

@article{BONOMO201559,
title = {Complexity of the cluster deletion problem on subclasses of chordal graphs},
journal = {Theoretical Computer Science},
volume = {600},
pages = {59-69},
year = {2015},
issn = {0304-3975},
doi = {https://doi.org/10.1016/j.tcs.2015.07.001},
url = {https://www.sciencedirect.com/science/article/pii/S0304397515005800},
author = {Flavia Bonomo and Guillermo Durán and Mario Valencia-Pabon}
}

@INPROCEEDINGS{9030927,
  author={Barr, Joseph R. and Shaw, Peter and Abu-Khzam, Faisal N. and Chen, Jikang},
  booktitle={2019 First International Conference on Graph Computing (GC)}, 
  title={Combinatorial Text Classification: the Effect of Multi-Parameterized Correlation Clustering}, 
  year={2019},
  volume={},
  number={},
  pages={29-36},
  doi={10.1109/GC46384.2019.00013}}

@InProceedings{10.1007/978-3-030-75242-2_15,
author="Figiel, Aleksander
and Himmel, Anne-Sophie
and Nichterlein, Andr{\'e}
and Niedermeier, Rolf",
editor="Calamoneri, Tiziana
and Cor{\`o}, Federico",
title="On 2-Clubs in Graph-Based Data Clustering: Theory and Algorithm Engineering",
booktitle="Algorithms  and Complexity",
year="2021",
publisher="Springer International Publishing",
address="Cham",
pages="216--230",
isbn="978-3-030-75242-2"
}

@article{BAZGAN2025247,
title = {On the hardness of problems around s-clubs on split graphs},
journal = {Discrete Applied Mathematics},
volume = {364},
pages = {247-254},
year = {2025},
issn = {0166-218X},
doi = {https://doi.org/10.1016/j.dam.2025.01.023},
url = {https://www.sciencedirect.com/science/article/pii/S0166218X25000290},
author = {Cristina Bazgan and Pinar Heggernes and André Nichterlein and Thomas Pontoizeau}
}

@InProceedings{10.1007/978-3-642-29700-7_22,
author="Liu, Hong
and Zhang, Peng
and Zhu, Daming",
editor="Snoeyink, Jack
and Lu, Pinyan
and Su, Kaile
and Wang, Lusheng",
title="On Editing Graphs into 2-Club Clusters",
booktitle="Frontiers in Algorithmics and Algorithmic Aspects in Information and Management",
year="2012",
publisher="Springer Berlin Heidelberg",
address="Berlin, Heidelberg",
pages="235--246",
isbn="978-3-642-29700-7"
}

@article{Huffner2010ClusterVertexDeletion,
  author    = {Falk H{\"u}ffner and
               Christian Komusiewicz and
               Hannes Moser},
  title     = {Fixed-Parameter Algorithms for Cluster Vertex Deletion},
  journal   = {Theory of Computing Systems},
  volume    = {47},
  number    = {1},
  pages     = {196--217},
  year      = {2010},
  doi       = {10.1007/s00224-008-9150-x},
  url       = {https://doi.org/10.1007/s00224-008-9150-x}
}

@article{fadiel2006computational,
  title={Computational Analysis of Mass Spectrometry Data Using Novel Combinatorial Methods.},
  author={Fadiel, Ahmed and Langston, Michael A and Peng, Xinxia and Perkins, Andy D and Taylor, Hugh S and Tuncalp, Ozge and Vitello, D and Pevsner, Paul H and Naftolin, Frederick},
  journal={AICCSA},
  volume={6},
  pages={8--11},
  year={2006}
}

@InProceedings{10.1007/11847250_2,
author="Dehne, Frank
and Langston, Michael A.
and Luo, Xuemei
and Pitre, Sylvain
and Shaw, Peter
and Zhang, Yun",
editor="Bodlaender, Hans L.
and Langston, Michael A.",
title="The Cluster Editing Problem: Implementations and Experiments",
booktitle="Parameterized and Exact Computation",
year="2006",
publisher="Springer Berlin Heidelberg",
address="Berlin, Heidelberg",
pages="13--24",
isbn="978-3-540-39101-2"
}

@INPROCEEDINGS{9253144,
  author={Barr, Joseph R. and Shaw, Peter and Abu-Khzam, Faisal N. and Yu, Sheng and Yin, Heng and Thatcher, Tyler},
  booktitle={2020 Second International Conference on Transdisciplinary AI (TransAI)}, 
  title={Combinatorial Code Classification \& Vulnerability Rating}, 
  year={2020},
  volume={},
  number={},
  pages={80-83},
  doi={10.1109/TransAI49837.2020.00017}}

@article{doi:10.1142/S1793351X20500087,
author = {Barr, Joseph R. and Shaw, Peter and Abu-Khzam, Faisal N. and Thatcher, Tyler and Yu, Sheng},
title = {Vulnerability Rating of Source Code with Token Embedding and Combinatorial Algorithms},
journal = {International Journal of Semantic Computing},
volume = {14},
number = {04},
pages = {501-516},
year = {2020},
doi = {10.1142/S1793351X20500087},

URL = { 
    
        https://doi.org/10.1142/S1793351X20500087
    
    

},
eprint = { 
    
        https://doi.org/10.1142/S1793351X20500087
    
    

}
}

@InProceedings{10.1007/978-3-642-16926-7_17,
author="Heggernes, Pinar
and Lokshtanov, Daniel
and Nederlof, Jesper
and Paul, Christophe
and Telle, Jan Arne",
editor="Thilikos, Dimitrios M.",
title="Generalized Graph Clustering: Recognizing (p,q)-Cluster Graphs",
booktitle="Graph Theoretic Concepts in Computer Science",
year="2010",
publisher="Springer Berlin Heidelberg",
address="Berlin, Heidelberg",
pages="171--183",
isbn="978-3-642-16926-7"
}

@article{GUO2009718,
title = {A more effective linear kernelization for cluster editing},
journal = {Theoretical Computer Science},
volume = {410},
number = {8},
pages = {718-726},
year = {2009},
issn = {0304-3975},
doi = {https://doi.org/10.1016/j.tcs.2008.10.021},
url = {https://www.sciencedirect.com/science/article/pii/S0304397508007822},
author = {Jiong Guo}
}

@InProceedings{10.1007/3-540-44849-7_17,
author="Gramm, Jens
and Guo, Jiong
and H{\"u}ffner, Falk
and Niedermeier, Rolf",
editor="Petreschi, Rossella
and Persiano, Giuseppe
and Silvestri, Riccardo",
title="Graph-Modeled Data Clustering: Fixed-Parameter Algorithms for Clique Generation",
booktitle="Algorithms and Complexity",
year="2003",
publisher="Springer Berlin Heidelberg",
address="Berlin, Heidelberg",
pages="108--119",
isbn="978-3-540-44849-5"
}

@article{Bocker2011Exact,
  author    = {Sebastian B{\"o}cker and
               Stefanie Briesemeister and
               Gunnar W. Klau},
  title     = {Exact Algorithms for Cluster Editing: Evaluation and Experiments},
  journal   = {Algorithmica},
  volume    = {60},
  number    = {2},
  pages     = {316--334},
  year      = {2011},
  doi       = {10.1007/s00453-009-9339-7},
  url       = {https://doi.org/10.1007/s00453-009-9339-7}
}

@article{ABUKHZAM2023113864,
title = {An improved fixed-parameter algorithm for 2-Club Cluster Edge Deletion},
journal = {Theoretical Computer Science},
volume = {958},
pages = {113864},
year = {2023},
issn = {0304-3975},
doi = {https://doi.org/10.1016/j.tcs.2023.113864},
url = {https://www.sciencedirect.com/science/article/pii/S0304397523001779},
author = {Faisal N. Abu-Khzam and Norma Makarem and Maryam Shehab}
}

@article{TSUR2022106171,
title = {Cluster deletion revisited},
journal = {Information Processing Letters},
volume = {173},
pages = {106171},
year = {2022},
issn = {0020-0190},
doi = {https://doi.org/10.1016/j.ipl.2021.106171},
url = {https://www.sciencedirect.com/science/article/pii/S0020019021000867},
author = {Dekel Tsur}
}

@InProceedings{cao_et_al:LIPIcs.IPEC.2021.13,
  author =	{Cao, Yixin and Ke, Yuping},
  title =	{{Improved Kernels for Edge Modification Problems}},
  booktitle =	{16th International Symposium on Parameterized and Exact Computation (IPEC 2021)},
  pages =	{13:1--13:14},
  series =	{Leibniz International Proceedings in Informatics (LIPIcs)},
  ISBN =	{978-3-95977-216-7},
  ISSN =	{1868-8969},
  year =	{2021},
  volume =	{214},
  editor =	{Golovach, Petr A. and Zehavi, Meirav},
  publisher =	{Schloss Dagstuhl -- Leibniz-Zentrum f{\"u}r Informatik},
  address =	{Dagstuhl, Germany},
  URL =		{https://drops.dagstuhl.de/entities/document/10.4230/LIPIcs.IPEC.2021.13},
  URN =		{urn:nbn:de:0030-drops-153965},
  doi =		{10.4230/LIPIcs.IPEC.2021.13}
}

@article{BOCKER20095467,
title = {Going weighted: Parameterized algorithms for cluster editing},
journal = {Theoretical Computer Science},
volume = {410},
number = {52},
pages = {5467-5480},
year = {2009},
note = {Combinatorial Optimization and Applications},
issn = {0304-3975},
doi = {https://doi.org/10.1016/j.tcs.2009.05.006},
url = {https://www.sciencedirect.com/science/article/pii/S0304397509003521},
author = {S. Böcker and S. Briesemeister and Q.B.A. Bui and A. Truss}
}

@article{SHAMIR2004173,
title = {Cluster graph modification problems},
journal = {Discrete Applied Mathematics},
volume = {144},
number = {1},
pages = {173-182},
year = {2004},
note = {Discrete Mathematics and Data Mining},
issn = {0166-218X},
doi = {https://doi.org/10.1016/j.dam.2004.01.007},
url = {https://www.sciencedirect.com/science/article/pii/S0166218X04001957},
author = {Ron Shamir and Roded Sharan and Dekel Tsur}
}

@book{marekcygan,
  author    = {Marek Cygan and
               Fedor V. Fomin and
               Lukasz Kowalik and
               Daniel Lokshtanov and
               D{\'{a}}niel Marx and
               Marcin Pilipczuk and
               Michal Pilipczuk and
               Saket Saurabh},
  title     = {Parameterized Algorithms},
  publisher = {Springer},
  year      = {2015}
 }

@book{Downey,
 author = {Downey, Rodney G. and Fellows, M. R.},
 title = {Parameterized Complexity},
 year = {2012},
 publisher = {Springer},
}
